\documentclass[a4paper,onecolumn,11pt,accepted=2021-06-16]{quantumarticle}
\pdfoutput=1

\usepackage{caption}
\usepackage{subcaption}

\usepackage{amsmath,amsfonts,amssymb,amsthm}
\usepackage[margin=1in]{geometry}
\usepackage[colorlinks]{hyperref}
\usepackage{xcolor}
\usepackage{mathrsfs}
\usepackage{mathtools}
\usepackage{thmtools}
\usepackage{thm-restate}
\mathtoolsset{centercolon}
\usepackage[numbers,comma,sort&compress]{natbib}
\usepackage[linesnumbered,ruled]{algorithm2e}
\usepackage{bm}
\usepackage{bookmark}
\usepackage{float}
\usepackage{multirow}
\usepackage[font=small]{caption}
\usepackage{graphicx}
\usepackage{epstopdf}

\usepackage{enumerate}

\let\originalleft\left
\let\originalright\right
\renewcommand{\left}{\mathopen{}\mathclose\bgroup\originalleft}
\renewcommand{\right}{\aftergroup\egroup\originalright}

\newcommand{\E}{{\mathbb{E}}}

\newcommand{\R}{{\mathbb{R}}}

\def\E{{\mathbb E}}
\def\V{{\mathbb V}}

\def\P{{\mathcal P}}

\renewcommand{\d}{\mathrm{d}}

\newcommand{\rangez}[1]{[{#1}]_0}

\newcommand{\be}{\begin{equation}}
\newcommand{\ee}{\end{equation}}
\newcommand{\bes}{\begin{equation*}}
\newcommand{\ees}{\end{equation*}}
\newcommand{\bea}{\begin{eqnarray}}
\newcommand{\eea}{\end{eqnarray}}
\newcommand{\beas}{\begin{eqnarray*}}
\newcommand{\eeas}{\end{eqnarray*}}

\newtheorem{theorem}{Theorem}
\newtheorem{lemma}{Lemma}
\newtheorem{corollary}{Corollary}
\newtheorem{problem}{Problem}
\newtheorem{proposition}{Proposition}

\newtheorem{assumption}{Assumption}

\numberwithin{equation}{section}
\allowdisplaybreaks
\newcommand{\eq}[1]{(\ref{eq:#1})}
\renewcommand{\sec}[1]{\hyperref[sec:#1]{Section~\ref*{sec:#1}}}
\newcommand{\app}[1]{\hyperref[app:#1]{Appendix~\ref*{app:#1}}}
\newcommand{\thm}[1]{\hyperref[thm:#1]{Theorem~\ref*{thm:#1}}}
\newcommand{\prop}[1]{\hyperref[prop:#1]{Proposition~\ref*{prop:#1}}}
\newcommand{\lem}[1]{\hyperref[lem:#1]{Lemma~\ref*{lem:#1}}}
\newcommand{\cor}[1]{\hyperref[cor:#1]{Corollary~\ref*{cor:#1}}}
\newcommand{\prb}[1]{\hyperref[prb:#1]{Problem~\ref*{prb:#1}}}
\newcommand{\tab}[1]{\hyperref[tab:#1]{Table~\ref*{tab:#1}}}
\newcommand{\ass}[1]{\hyperref[ass:#1]{Assumption~\ref*{ass:#1}}}
\newcommand{\fig}[1]{\hyperref[fig:#1]{Figure~\ref*{fig:#1}}}

\makeatletter
\let\@@magyar@captionfix\relax
\makeatother

\begin{document}

\title{Quantum-accelerated multilevel Monte Carlo \newline methods for stochastic differential equations in mathematical finance}

\author[1]{Dong An}
\author[2]{Noah Linden} 
\author[3,4,5]{Jin-Peng Liu}
\author[2,6]{Ashley Montanaro}
\author[2]{Changpeng Shao}
\author[1]{Jiasu Wang}
\affil[1]{Department of Mathematics, University of California, Berkeley, CA 94720, USA}
\affil[2]{School of Mathematics, Fry Building, University of Bristol, BS8 1UG, UK}
\affil[3]{Joint Center for Quantum Information and Computer Science, University of Maryland, MD 20742, USA}
\affil[4]{Institute for Advanced Computer Studies, University of Maryland, MD 20742, USA}
\affil[5]{Department of Mathematics, University of Maryland, MD 20742, USA}
\affil[6]{Phasecraft Ltd, Quantum Technologies Innovation Centre, Bristol BS1 5DD, UK}

\maketitle

\begin{abstract}
  Inspired by recent progress in quantum algorithms for ordinary and partial differential equations, we study quantum algorithms for stochastic differential equations (SDEs). 
Firstly we provide a quantum algorithm that gives a quadratic speed-up for multilevel Monte Carlo methods in a general setting. As applications, we apply it to compute expectation values determined by classical solutions of SDEs, with improved dependence on precision.
We demonstrate the use of this algorithm in a variety of applications arising in mathematical finance, such as the Black-Scholes and Local Volatility models, and Greeks. We also provide a quantum algorithm based on sublinear binomial sampling for the binomial option pricing model with the same improvement.
\end{abstract}

\section{Introduction}
\label{sec:introduction}

Differential equations are ubiquitous throughout mathematics, science, and engineering. Specifically, ordinary differential equations (ODEs) and partial differential equations (PDEs) characterize continuous processes of systems arising extensively in many fields, from solid mechanics, fluid dynamics, and electromagnetism to biology \cite{Eva10}. Calculations of properties of such deterministic systems, typically require numerical schemes, that is one discretizes the differential equation in order to provide an approximate value of the quantity of interest \cite{Atk08}. 

For numerous systems arising in statistical physics, molecular dynamics, finance, and other real-world models, the dynamics  is captured by a stochastic differential equation (SDE) \cite{KP13,mao2007stochastic}. 
 Given a typical SDE, a fundamental computational problem is to provide an expected value of a random variable $Y$, denoted $\E[Y]$, which is a functional determined by the solution of the SDE.
Such a computational problem has been widely studied in mathematical finance, where the quantity $Y$ represents the payoff in option and derivative pricing. 
It is often computationally expensive to estimate $\E[Y]$, since a scheme that approximates the SDE is necessarily run many times to average over the randomness. 
In this domain, Monte Carlo (MC) methods are basic tools with a provable complexity analysis.

Monte Carlo simulation, known for its flexibility and generality, refers to performing Monte Carlo methods in simulating SDEs \cite{KP13,KW09}. This method makes use of randomness to estimate $\E[Y]$ of a SDE as introduced above. In general, Monte Carlo simulation generates $k$ independent approximate samples from $Y$ by performing a chosen scheme, and then outputs an average of those $k$ outputs as an approximate expectation of $Y$. Assuming the variance of $Y$ is bounded by $\sigma^2$, according to Chebyshev's inequality, it suffices to utilize $k=O(\sigma^2/\epsilon^2)$ samples to estimate $\E[Y]$, where $\epsilon$ is the additive error \cite{KW09}. 

Simulating SDEs using Monte Carlo is very useful when the cost of each sample is cheap. A typical example is the Geometric Brownian Motion (GBM) in the Black-Scholes (BS) model \cite{Bod09,BS73,Hul03}, which characterizes stock prices in the financial market. Classical algorithms are designed to directly sample from the solution of GBM, whose explicit form is known, rather than simulating the paths of the GBM itself. Assuming the cost of sampling the solution is ignorable, i.e. $O(1)$, the computational complexity equals the number of samples, $O(1/\epsilon^2)$. However, a challenge faced by Monte Carlo occurs when creating each sample is costly, so even quadratic dependence on $\epsilon$ may result in an  computational complexity significantly larger than $O(1/\epsilon^2)$ in practice. More concretely, we consider a general SDE without an explicit solution. To compute an approximate solution by numerical methods, we discretize the SDE on the time interval $[0,T]$ with the step size $h$, and perform a scheme of strong order $r$ (which produces a random variable $\widehat{Y}$ that is $\epsilon$-close to $Y$, where $\epsilon = O(h^r)$) \cite{KP13}. Taking the Milstein scheme with strong order $1$ as an example, \emph{i.e.} $r=1$, then $T/h = \Omega(1/\epsilon)$ number of iterations is required to produce one sample. Under this circumstance, the computational complexity of classical Monte Carlo simulation is $O(1/\epsilon^3)$ in total. Generally, when performing a scheme of strong order $r$, the complexity can be improved to $O(1/\epsilon^{2+1/r})$. Usually, it is challenging to perform a scheme with  large $r$, due to the higher smoothness requirement of the SDE, and, in practice, it is harder to implement explicit forms of higher strong order schemes \cite{BBT04,KP13,Kuz18}. So the acceleration is moderate in practice. 

To reduce such an expensive computational cost, multilevel Monte Carlo (MLMC) methods have attracted very considerable attention recently, and have successfully been applied to simulate SDEs with applications in finance \cite{Gil08,Gil15}. Recalling the goal of estimating $\E[Y]$ of a general random variable $Y$, given a sequence of estimators $P_0, P_l, \ldots, P_L$ that approximates $Y$ with increasing accuracy and cost, multilevel Monte Carlo aims to estimate $\E[P_L]$ by simulating a sum of $\E[P_l-P_{l-1}]$, with different numbers of samples at each level $l$.
The main improvement of multilevel Monte Carlo is from reducing the number of samples when the variance of $P_l-P_{l-1}$ is large. 
By balancing the sample numbers and variances for different $P_l-P_{l-1}$ and summing them together, multilevel Monte Carlo gives an optimal overall cost of estimating $\E[P_L]$ that approximates $\E[Y]$ within the mean-squared error $\epsilon^2$. 
As described by Giles et al, multilevel Monte Carlo with a scheme of strong order $r>1$, is capable of estimating $\E[Y]$ of general SDEs with the overall cost $\widetilde O(1/\epsilon^2)$ (\cite[Theorem 1]{Gil15}), where $\widetilde O$ neglects logarithmic factors. Compared to the standard Monte Carlo simulation with the same scheme, it removes a $1/\epsilon^{1/r}$ 
factor in the overall complexity. Moreover, this approach does not require the use of higher strong order schemes, and hence avoids the smoothness requirement and the implementation difficulty as well. 
We note that when the samples are allowed not to be random and independent, but based on certain lattice rules, it could lead to a $O(1/\epsilon^p)$ cost with $p<2$ under certain conditions \cite{DKS13,GW09}. In our case we assume the samples are chosen randomly and independently, then $O(1/\epsilon^2)$ is the best known complexity that classical algorithm can achieve. \cite{Gil15} 
However such a quadratic dependence on $1/\epsilon$ of the overall complexity that multilevel Monte Carlo can achieve is still far from ideal, particularly when applications to mathematical finance are considered. 

Quantum computers are expected to outperform classical computers for solving a system of linear equations \cite{Amb12,AL19,CKS15,GSLW18,HHL08,LT19,SSO18,TAWL20} and differential equations \cite{Ber14,BCOW17,CL19,CJS13,CPPTK13,MP16,CJO19,CLO20,ESP19,LKK20,LMS20,XDGS18,XDGS19}. Quantum algorithms for certain stochastic differential equations, such as simulating GBM of the Black-Scholes model as discussed above, have  attracted increasing attention in quantum computational finance \cite{BDJ20,FJO19,OML19,RGB18, GRSS21}. 
Reference~\cite{GRSS21} claims an exponential speedup over classical algorithms to solve the Black-Scholes PDE, but does not include a detailed complexity analysis; it uses a very different approach to that used here, which does not seem to be easily extendible to general SDEs.
For the payoff models, using quantum-accelerated Monte Carlo methods \cite{montanaro2015quantum} based on sampling from the explicit solution of GBM, quantum algorithms can approximate the expected value of the price of portfolio within error $\epsilon$ in complexity $\widetilde{O}(1/\epsilon)$,
a quadratically improved dependence on $\epsilon$ compared to classical Monte Carlo simulation \cite{KMTY21,OML19,RGB18,SES20}. However, previous quantum algorithms, that sample the explicit solution of GBM, cannot be extended to simulate general SDEs with no explicit formula for the solution. Reference \cite{KMTY20} presented a practical quantum circuit for simulating the Local Volatility (LV) model, which generalizes the GBM and does not have an explicit solution \cite{DK94,Dup94}. However, reference \cite{KMTY20} did not provide a concrete complexity for simulating the LV model. In summary, quantum speedups for payoff models of general SDEs lacking explicit solutions are far from well established.

In this paper, we provide quantum algorithms for approximating classical outputs determined by general SDEs (i.e. ones with no explicit formulas for the solutions) with computational complexity $\widetilde O(1/\epsilon)$. We will apply these to several payoff models of SDEs which arise in finance, and the outputs are the payoffs in pricing. To achieve such an improvement, we first propose a quantum-accelerated multilevel Monte Carlo (QA-MLMC) method in a general setting, and then apply QA-MLMC for general SDEs. Compared to the classical counterpart, QA-MLMC achieves a quadratic speedup in precision up to a logarithmic factor. The main ingredient to this acceleration is the quantum speedup of the Monte Carlo method \cite{montanaro2015quantum}. Roughly, to approximate $\E[P_l - P_{l-1}]$ in the MLMC approach as we discussed above, we only need to use $\widetilde{O}(1/\epsilon)$ samples. We shall prove that this speedup is preserved in the telescoping sum $\sum_{l=0}^L \E[P_l - P_{l-1}]$ to approximate $\E[P_L]$. 
We remark that a somewhat similar idea was used in~\cite{montanaro2015quantum} for the special case of computing a partition function as a telescoping \emph{product} of terms, each of which is approximated using QA-MC.

Instead of the mean-squared error considered in the classical case, QA-MLMC only returns an approximation of $\E[P_L]$ in the sense of the additive error. However, this will not incur any unfair comparison between MLMC and QA-MLMC since the two types of errors are almost equivalent (see \app{errors}). Because of the different types of errors, MLMC and QA-MLMC have slight differences in their assumptions. 
This will become clear below when we apply QA-MLMC to solve practical problems in finance.

As discussed above, to solve a general SDE with no explicit solutions, usually we need a discretization scheme of strong order $r$. Classically, to apply MLMC, a numerical scheme of strong order $r>1$ is usually sufficient to obtain a complexity of $\widetilde{O}(1/\epsilon^2)$. However, in the quantum case, to ensure a quadratic speedup, i.e. to obtain a complexity of $\widetilde{O}(1/\epsilon)$, we have to apply a numerical scheme of strong order $r>2$. 
Note that when an order $r$ scheme is used, the Monte Carlo method can solve the SDE with a complexity of $\widetilde{O}(1/\epsilon^{2+1/r})$. And a direct corollary of \cite{montanaro2015quantum} shows that in the quantum case this method can be improved to has a complexity of $\widetilde{O}(1/\epsilon^{1+1/r})$. However, for the reasons discussed above we cannot choose $r$ as large as we want. So QA-MLMC still has advantages over QA-MC both in theory and in practice.

\begin{table}[ht]
\begin{center}
\footnotesize{
\renewcommand{\arraystretch}{1.25}
\begin{tabular}{|@{\hspace{1.5mm}}c|c|c|c|c|}
    \hline
     & \textbf{Algorithm} & \textbf{Model} & \textbf{Result} \\
    \hline
    \parbox[t]{1.5mm}{\multirow{4}{*}{\rotatebox[origin=c]{90}{Classical}}}
     & MC with direct sampling \cite{Bod09,Hul03} & Black-Scholes model & $\epsilon^{-2}$ \\
     \cline{2-4}
     & MC with scheme of strong order $r$ (\prop{classical_MC}) & payoff models of general SDEs & $\epsilon^{-2-1/r}$ \\
     \cline{2-4}
     & MLMC with scheme of strong order $r>1$ (\prop{MLMC_SDE})  & payoff models of general SDEs & $\epsilon^{-2}$ \\
    \cline{2-4}
     & MC with binomial sampling (\prop{BOPM}) & binomial option pricing model & $\epsilon^{-2}$ \\
    \hline
    \parbox[t]{1.5mm}{\multirow{4}{*}{\rotatebox[origin=c]{90}{Quantum}}}
     & QA-MC with direct sampling \cite{OML19,RGB18,SES20} & Black-Scholes model & $\epsilon^{-1}$ \\
    \cline{2-4}
    & QA-MC with scheme of strong order $r$ (\thm{quantum_MC})  & payoff models of general SDEs & $\epsilon^{-1-1/r}$ \\
    \cline{2-4}
     & QA-MLMC with scheme of strong order $r>2$ (\thm{QMLMC_SDE})  & payoff models of general SDEs & $\epsilon^{-1}$ \\
    \cline{2-4}
     & QA-MC with binomial sampling (\thm{QBOPM})  & binomial option pricing model & $\epsilon^{-1}$ \\
    \hline
\end{tabular}
}
\end{center}
\caption{
Summary of the time complexities of classical and quantum algorithms for financial models with the additive error $\epsilon$, in which logarithmic factors are omitted. 
\label{tab:alg-compare}
}
\end{table}

As applications, we apply QA-MLMC to solving various payoff models, which satisfy our smoothness requirements and are of great interest in mathematical finance. Examples include the well-known Black-Scholes model that prices a variety of financial derivatives; the Local Volatility model that generalizes the Black-Scholes model by treating volatility as a function of the asset and the time; the Greeks that label the sensitivity of the price of the option for hedge portfolios \cite{Hul03}; and the binomial option pricing model as introduced above. 

For the analytically solvable Black-Scholes model, in which QA-MC has been applied to reduce the complexity from $\widetilde O(\epsilon^{-2})$ to $\widetilde O(\epsilon^{-1})$ \cite{RGB18}, we verify QA-MLMC is able to achieve the same quantum speedup. For the rest of the models, we establish the first quantum acceleration to achieve the complexity $\widetilde O(\epsilon^{-1})$, by applying QA-MLMC to the Local Volatility model and Greeks.
\tab{alg-compare} compares the performance of our approaches to classical ones for various financial models with respect to the dependence on error tolerance $\epsilon$.  

Furthermore, we also study the Black-Scholes model with European and digital payoffs in numerical experiments.
We provide concrete numerical implementations of the schemes of strong order up to $3$ to test those parameters with different payoff functions, and numerical results are in good agreement with our theoretical estimates. This provides evidence that our complexities for MLMC and QA-MLMC are reasonable and sharp.

We also consider the binomial option pricing model (BOPM), also known as the binomial lattice model (BLM), which provides an alternative option pricing model different from the Black-Scholes model, by constructing a binomial tree with the same expectation and variance as the Geometric Brownian Motion \cite{CRR79,Hul03,Ren79}. Inspired by the binomial structure, it is natural to simulate BOPM by performing sublinear binomial sampling, a recently developed technique that samples a binomial tree in sublinear time \cite{BKP15,FT15}. This technique has been applied to construct a fast random walk, which is a specific binomial tree, to develop fast classical and quantum algorithms for heat equation \cite{LMS20}. Observing that sublinear binomial sampling provides a numerical method to output an estimate of $Y$ with ignorable cost per sample, we can estimate $\E[Y]$ preserving the computational complexity the same as the sampling complexity, neglecting logarithmic factors. Therefore, we propose classical and quantum algorithms for BOPM, with the dependence on precision $\widetilde O(1/\epsilon^2)$ and $\widetilde O(1/\epsilon)$, respectively. Complementing the multilevel Monte Carlo methods, this provides an alternative method for efficiently studying properties of structured stochastic models in mathematical finance.

Overall, our theoretical and numerical results provide promise for potential applications of quantum computing in computational finance.

The rest of this paper is structured as follows. \sec{problem} introduces the payoff problem of SDEs that we study, and quantum-accelerated Monte Carlo for solving SDEs. \sec{QMLMC} states the main theorem for the quantum-accelerated multilevel Monte Carlo method. \sec{SDE} applies the quantum-accelerated multilevel Monte Carlo method to the SDE problem. \sec{BSM} presents the Black-Scholes model as an application for estimating the price of the option. \sec{LVM} generalizes the Black-Scholes model to the Local Volatility model. \sec{Greeks} introduces the Greeks as an application for estimating the sensitivity of the price. \sec{BLM} covers the binomial lattice model as an alternative option pricing model. \sec{discussion} includes a discussion and raises some open problems. 
\app{errors} discusses the relationship between mean square error and additive error.
Finally, \app{numerical results} presents numerical results of several schemes for the Black-Scholes model with European and digital payoffs.

\section{Payoff models of general SDEs}
\label{sec:problem}

\subsection{Problem settings}

Suppose that we have a stochastic differential equation (SDE) with general drift and volatility terms
\begin{equation}
\d{X_t} = \mu(X_t,t)\d t + \sigma(X_t,t)\d W_t,
\label{eq:SDE}
\end{equation}
for $t\in[0,T]$, where $X_t\in\R$ is an It\^{o} process, $W_t$ is a standard Brownian motion. The payoff problem we are concerned with is as follows.

\begin{problem}\label{prb:problem}
Assuming there exists an oracle $O_I$ that samples from an initial distribution $\pi_0$ to produce $X_0$ as an initial condition, an oracle $O_W$ that samples from the Brownian motion $W_t$, and an oracle $O_{\P}$ that produces the payoff $\P(X)$ as a functional of a specific $X$, and oracles $O_{\mu}$ and $O_{\sigma}$ that produce $\mu(X,t)$ and $\sigma(X,t)$ as functions of $X$ and $t$, respectively.

Given an evolution time $T>0$, we aim to compute
\begin{equation}
\E[\P(X_T) ~|~ X_0\in\pi_0],
\end{equation}
within an error $\epsilon>0$, where $X_T$ is generated by \eq{SDE}.
\end{problem}

\prb{problem} is widely investigated in mathematical finance. Different kinds of financial models, such as option pricing and Greeks, can be formulated as \prb{problem} with various assumptions, which are introduced in detail in \sec{SDE}. Note that if $\pi_0=\delta_{X_0}$,
then $X\in\pi_0 \to X=X_0$, and \prb{problem} is reduced to a deterministic initial value problem. 

Before we proceed, it is helpful to clarify the meaning of the errors. 
By saying to compute $\E[\P]$ within an error $\epsilon > 0$ in \prb{problem}, there are two different scenarios (here the random variable $Y$ is an estimator of $\E[\P]$): 
\begin{itemize}
    \item the mean-squared error $\E(Y-\E[\P])^2$ is bounded by $\epsilon^2$,
    \item the additive error $|Y-\E[\P]|$ is bounded by $\epsilon$ with probability at least 0.99.
\end{itemize}
Most research on classical Monte Carlo simulation uses mean-squared error, whereas research on the quantum-accelerated Monte Carlo method typically uses additive error. 
To reduce technical difficulty and be consistent with existing literature, we will follow this convention, using mean-squared error in our classical algorithms and using additive error in our quantum algorithms. 
We note however that this still allows fair comparison between classical and quantum algorithms, because these two types of error bounds are indeed almost equivalent, which is elaborated in detail in \app{errors}. 

\subsection{Monte Carlo method}

For simplicity we assume that the costs of querying $O_I$, $O_W$, $O_{\P}$, $O_{\mu}$ and $O_{\sigma}$ are $O(1)$. For our quantum algorithms, we assume that the oracles 
can produce coherent superpositions corresponding to these distributions. This assumption can usually be satisfied, for example given a classical algorithm that generates samples based on uniformly random bits; see~\cite{montanaro2015quantum} for a discussion. In particular, note that we will not actually need to put the initial distribution (which may be quite complicated) into superposition -- the algorithm will compute $\E[\P(X_T)]$ for a given starting position $X_0$, and sampling $X_0$ can be performed classically. If a sample of $X_T$ generated by \eq{SDE} can be obtained with cost $O(1)$, the computational complexity of solving \prb{problem} equals the total number of samples required to estimate $\E[\P(X_T)]$ within $\epsilon$. While for general SDEs, the cost of simulating \eq{SDE} to calculate $X_T$ each time should be taken into account. 

For instance, we consider the widely used Milstein discretization with the step size $h$, giving
\begin{equation}
\widehat{X}_{k+1} = \widehat{X}_k + \mu(\widehat{X}_k,t_k)h + \sigma(\widehat{X}_k,t_k) \Delta W_k + \frac{1}{2}\sigma(\widehat{X}_k,t_k)\partial_X \sigma(\widehat{X}_k
,t_k)((\Delta W_k)^2-h),
\end{equation}
where $k\in\rangez{n}$, where $n = T/h$. This has computational cost $O(1/h)$ to produce a sample $\widehat{X}_n$ that approximates $X_T$. Based on Monte Carlo methods \cite{KP13,KW09}, a simple estimation for $\E[\P(X_T)]$ is to generate $N$ samples, each independently outputs $\P(X_T)$, and then to produce an average,
\be  \label{Expection for introducing MLML}
Y = \frac{1}{N} \sum_{i=1}^N \P(\widehat{X}_{n}^{(i)}).
\ee
Note that the mean-squared error can be decomposed as 
\begin{equation}\label{eq:meansquared_decom}
    \E|Y - \E[\P(X_T)]|^2 = \V [Y] + |\E [Y] - \E[\P(X_T)]|^2,
\end{equation}
which can be bounded by $O(N^{-1} + h^2)$. Here we let $\V[Y]$ denote the variance of a random variable $Y$.
If we seek to bound the mean-squared error by $\epsilon^2$, then we can take $N=O(1/\epsilon^2)$, $h = \Omega(\epsilon)$, so the expected computational cost equals $O(N/h) = O(\epsilon^{-3})$.
This result is first stated in \cite{Gil08}. In general, for high order schemes, the complexity of classical Monte Carlo method can be bounded as follows. 
\begin{proposition}\label{prop:classical_MC}
    Let $Y$ be an estimator of $\E[\P(X_T)]$ with bounded variance, based on a numerical discretization with time step size $h$ using a numerical scheme with strong order $r$ (defined in \eq{def_strong_r}). 
    Then in order to achieve the accuracy of mean-squared error $\epsilon^2$, the computational complexity of the Monte Carlo method is $O(\epsilon^{-2-1/r})$. 
\end{proposition}
\begin{proof}
    For simplicity we assume $\P$ is globally Lipschitz continuous (we refer to \sec{SDE} for a more general analysis).  
    To bound the mean-squared error, we only need to bound the right hand side of the equation \eq{meansquared_decom}. 
    Note that 
    \begin{equation}
    \V [Y] = \frac{1}{N} \V[\P(\widehat{X}_T)] = O(N^{-1}),
    \end{equation}
    and 
    \begin{equation}
    |\E [Y] - \E[\P(X_T)]|^2 = |\E [\P(\widehat{X}_n)] - \E[\P(X_T)]|^2 \leq \E |\P(\widehat{X}_n) - \P(X_T)|^2 = O(h^{2r}). 
    \end{equation}
    Therefore 
    \begin{equation}
        \E|Y - \E[\P(X_T)]|^2 =  O(N^{-1}+h^{2r}). 
    \end{equation}
    In order to bound the mean-squared error by $\epsilon^2$, it suffices to choose $N \sim \epsilon^{-2}$ and $h \sim \epsilon^{1/r}$, thus the complexity becomes $O(N/h) = O(\epsilon^{-2-1/r}).$
\end{proof}

\prop{classical_MC} tells us that in the classical Monte Carlo method, the total complexity can be indeed improved by using a higher order scheme, but it is bottlenecked by $O(\epsilon^{-2})$ due to the variance of the estimator. 
To minimize the overall computational cost without needing higher order schemes, we will introduce an advanced Monte Carlo approach, known as multilevel Monte Carlo \cite{Gil08}, in \sec{QMLMC}.

However before doing so, we show that quantum-accelerated Monte Carlo, as described in \cite{montanaro2015quantum} can give a speed up of non-multilevel methods.

\subsection{Quantum-accelerated Monte Carlo method}

In \cite{montanaro2015quantum}, Montanaro showed that using a quantum computer, the number of samples used in Monte Carlo can be reduced quadratically.

\begin{lemma}[Theorem 5 of \cite{montanaro2015quantum}]
\label{lem:Montanaro2015}
Let $\mathcal{A}$ be a (classical or quantum) algorithm. Let $v(\mathcal{A})$ be the random variable corresponding to $v(x)$ when the outcome of $\mathcal{A}$ is $x$.
Assume that $\V[v(\mathcal{A})]\leq \sigma^2$, then there is a quantum algorithm that estimates $\E[v(\mathcal{A})]$ up to additive error $\epsilon$ with success probability at least $2/3$
by using
\be
O\Bigl((\sigma/\epsilon)(\log \sigma/\epsilon)^{3/2} (\log\log \sigma/\epsilon)\Bigr)
\ee
samples. 
\end{lemma}

The powering lemma stated below can increase the success probability of \lem{Montanaro2015} to $1-\delta$ for any arbitrarily small $\delta$.

\begin{lemma}[Lemma 1 of \cite{montanaro2015quantum}]
\label{lem:powering lemma}
Let $\mathcal{A}$ be a (classical or quantum) algorithm which aims to estimate some quantity $\mu$, and whose output $\tilde{\mu}$ satisfies $|\mu - \tilde{\mu}|\leq \epsilon$ except with probability $\gamma$, for some fixed $\gamma<1/2$. Then, for any $\delta>0$, it suffices to repeat $\mathcal{A}$ $O(\log 1/\delta)$ times and take the median to obtain an estimate which is accurate to within $\epsilon$ with probability at least $1 - \delta$.
\end{lemma}


In \lem{Montanaro2015} and \lem{powering lemma}, the randomized or quantum algorithm $\mathcal{A}$ is used to produce a random variable $X$ and then compute the payoff function $\P(X)\in\R$. For the detailed implementation of $\mathcal{A}$, we refer to the start of Section 2 of \cite{montanaro2015quantum} for more details.

Based on \lem{powering lemma}, the success probability of \lem{Montanaro2015} can be improved to $1-\delta$ by using $O((\sigma/\epsilon)(\log \sigma/\epsilon)^{3/2} (\log\log \sigma/\epsilon)(\log 1/\delta))$ samples. Thus, we are able to develop a quantum-accelerated Monte Carlo method for \prb{problem}.

\begin{theorem}\label{thm:quantum_MC}
    Let $\mathcal{A}$ be an algorithm that generates a sample of numerical solution $\widehat{X}_n$ of the SDE using a numerical discretization with time step size $h$ using $r$-th order scheme in the sense that $\E|\widehat{X}_n-X_T| = O(h^{r})$. 
    Assume that $\P(\widehat{X}_n)$ has bounded variance independent of $h$. 
    Then there exists a quantum algorithm that achieves the accuracy of additive error $\epsilon$ with probability at least $0.99$, with computational complexity $\widetilde O(\epsilon^{-1-1/r})$. 
\end{theorem}
\begin{proof}
    Similarly to \prop{classical_MC}, for technical simplicity we assume $\P$ is globally Lipschitz continuous (we refer to \sec{SDE} for analysis on more general payoff functions). By \lem{Montanaro2015} and \lem{powering lemma}, there exists a quantum algorithm that generates an estimator $Y$ such that $$|Y - \E [\P(\widehat{X}_n)]| < \epsilon/2$$ with probability at least 0.99 using $\widetilde{O}(\epsilon^{-1})$ queries to $\mathcal{A}$. 
    Furthermore, 
    \begin{equation}
    |\E [\P(\widehat{X}_n)] - \E[\P(X_T)]| \le \E |\P(\widehat{X}_n) - \P(X_T)| = O(h^{r}). 
    \end{equation}
    Hence we can choose $h = O(\epsilon^{-1/r})$ to bound $|\E [\P(\widehat{X}_n)] - \E[\P(X_T)]|$ by $\epsilon/2$, and the complexity of each query to $\mathcal{A}$ becomes $O(\epsilon^{-1/r})$. 
    Combining the above two estimates, 
    in order to bound the additive error $|Y - \E[\P(X_T)]|$ by $\epsilon$ with probability at least 0.99, the total complexity is $\widetilde{O}(\epsilon^{-1-1/r})$. 
\end{proof}

\section{Quantum-accelerated MLMC}
\label{sec:QMLMC}

\subsection{Multilevel Monte Carlo method}

Heinrich \cite{Hei01} developed the first work on multilevel Monte Carlo (MLMC) methods for parametric integration, then Giles \cite{Gil08} introduced MLMC to simulate SDEs. 
In the following, we first briefly introduce this method. Then we show how to accelerate this method using a quantum-accelerated Monte Carlo method \cite{montanaro2015quantum}. For more about MLMC, we refer to the survey paper \cite{Gil15}. 

Instead of focusing on SDEs, we review the idea of MLMC in the general setting.
Now let $P$ be a random variable, our goal is to estimate $\E[P]$, given a sequence $P_0,P_1,\ldots,P_{L}$ that approximates $P$ with increasing accuracy, but also increasing cost. For instance, for the SDE \eq{SDE}, $P$ is the payoff, $P_l=\P(\widehat{X}_{n_l})$ with $n_l=T/h_l=2^{l}T$.
Now we have the following telescoping sum
\be \label{telescoping sum}
\E[P_L] = \sum_{l=0}^L \E[P_l-P_{l-1}],
\ee
where $P_{-1}=0$. We can estimate $\E[P_L]$ by using the Monte Carlo method to approximate each term $\E[P_l-P_{l-1}]$. So we obtain the following approximation of $\E[P_L]$
\be \label{approximation of EPL}
Y = \sum_{l=0}^L  Y_l, \quad {\rm where}~Y_l := \frac{1}{N_l}
\sum_{i=0}^{N_l} \left(P_l^{(l,i)} - P_{l-1}^{(l,i)} \right) .
\ee
The superindex $l$ means that the samples are generated independently.

Note that for SDE, $P_l - P_{l-1}$ comes from two discrete approximations with different timesteps but the same Brownian path. To generate random samples $P_l^{(l,i)} - P_{l-1}^{(l,i)}$, one method suggested in \cite{Gil08} is as follows. First constructing the Brownian increments for the simulation of the discrete path leading to the evaluation of $P_l^{(l,i)}$. Then summing them in groups of size $2$ to give the discrete Brownian increments for the evaluation of $P_{l-1}^{(l,i)}$.

To determine the cost, the MLMC approach considers the mean-squared error, 
\be
\E(Y-\E[P])^2,
\ee
which can be decomposed as follows
\be \label{mean-square error}
\E(Y-\E[P])^2 = \V [Y] + (\E[P_L] - \E[P])^2 \leq \V [Y] + \E[P_L - P]^2.
\ee
In order to achieve $\E (Y-\E[P])^2\le\epsilon^2$, it is sufficient to ensure that $\V[Y]\leq \epsilon^2/2$ and $\E[P_L-P]^2 \leq \epsilon^2/2$.
For $l\geq 0$, let $C_l,V_l$ be the cost and variance of one sample of $P_l-P_{l-1}$. Then
the overall cost and variance of $Y$ given in (\ref{approximation of EPL}) is
$\sum_{l=0}^L N_lC_l$ and $\sum_{l=0}^L N_l^{-1} V_l$.

To minimize the cost with a fixed variance $\epsilon^2/2$, we introduce the Lagrange multiplier $\lambda^2$ and minimize 
\be
\sum_{l=0}^L N_lC_l + \lambda^2 \left( \frac{\epsilon^2}{2} - \sum_{l=0}^L N_l^{-1} V_l \right).
\ee
This leads to $N_l = \lambda \sqrt{V_l/C_l}$
and $\lambda = 2\epsilon^{-2} \sum_{l=0}^L \sqrt{V_lC_l}$. The total computational cost
is
\be
2\epsilon^{-2} \left(\sum_{l=0}^L \sqrt{V_lC_l} \right)^2.
\ee

We restate the rigorous argument in \cite{Gil15} as follows.

\begin{lemma}[Theorem 1 of \cite{Gil15}]
\label{lem:MLMC}
Let $P$ be a random variable and $P_l$ be the corresponding level $l$ numerical approximation.
Let $Y_l$ be an approximation of $\E[P_l-P_{l-1}]$ based on Monte Carlo method such that the expected cost and variance of one sample is $C_l$ and $V_l$ respectively.
If there exist positive constants $\alpha,\beta,\gamma$ such that 
{$\alpha\geq \frac{1}{2} \min(\beta,\gamma)$} 
and
\begin{itemize}
\item $|\E[P_l-P]| = O(2^{-\alpha l})$,
\item $\E[Y_l] = \E[P_l-P_{l-1}],~l\geq 0$, where $P_{-1}=0$,
\item {$V_l= O(2^{-\beta l})$},
\item $C_l= O(2^{\gamma l})$,
\end{itemize}
then for any $\epsilon< 1/e$ there exists an $L$ such that
$Y=\sum_{l=0}^L Y_l$
has a mean-squared error with bound
{$\E(Y - \E[P])^2 \leq \epsilon^2$}. Moreover, the total computational cost is
\be
\begin{cases}
O(\epsilon^{-2}), & \beta>\gamma, \\
O(\epsilon^{-2}(\log\epsilon)^2), & \beta=\gamma, \\
O(\epsilon^{-2-(\gamma-\beta)/\alpha}), & \beta<\gamma.
\end{cases}
\ee
\end{lemma}

\subsection{Quantum-accelerated multilevel Monte Carlo method}

Recall that in MLMC, in the telescoping sum (\ref{telescoping sum}), the Monte Carlo method is utilized to approximate each mean value $\E[P_l- P_{l-1}]$, in which we can obtain a quadratic speedup using a quantum computer. To approximate the telescoping sum (\ref{telescoping sum}) via  (\ref{approximation of EPL}), we have a classical algorithm to do the sampling, and thus the assumption of the existence of the classical algorithm to do the sampling is satisfied for MLMC in \lem{Montanaro2015}.
The proof of the following theorem is similar to that of \lem{MLMC}.

\begin{theorem}
\label{thm:QMLMC}
Let $P$ denote a random variable, and let $P_l~(l=0,1,\ldots,L)$ denote a sequence of random variables such that 
$P_l$ approximates $P$ at level $l$. Further define $P_{-1}=0$.
Let $C_l$ be the cost of sampling from $P_l$, and let $V_l$ be the variance of $P_l-P_{l-1}$.
If there exist positive constants $\alpha,
\beta = 2 \hat{\beta},\gamma$
such that 
$\alpha\geq \min(\hat{\beta},\gamma)$
and
\begin{itemize}
\item $|\E[P_l-P]| = O(2^{-\alpha l})$,
\item { $V_l= O(2^{-\beta l}) = O( 2^{-2\hat{\beta} l})$},
\item $C_l = O(2^{\gamma l})$,
\end{itemize}
then for any $\epsilon< 1/e$ 
there is a quantum algorithm that 
estimates $\E[P]$ up to additive error
$\epsilon$ with probability at least 0.99, and with cost
\be
\begin{cases} \vspace{.1cm}
O\Bigl(\epsilon^{-1} (\log 1/\epsilon)^{3/2} (\log\log 1/\epsilon)^2\Bigr), & \hat{\beta}>\gamma, \\ \vspace{.1cm}
O\Bigl(\epsilon^{-1}(\log 1/\epsilon)^{7/2} (\log\log  1/\epsilon)^2\Bigr)
, & \hat{\beta}=\gamma, \\
O\Bigl(\epsilon^{-1-(\gamma-\hat{\beta})/\alpha} (\log 1/\epsilon)^{3/2} (\log\log 1/\epsilon)^2\Bigr), & \hat{\beta}<\gamma.
\end{cases}
\ee
\end{theorem}

\begin{proof}
In MLMC, we use the telescoping sum $\E[P_L] = \E[P_0] + \sum_{l=1}^L \E[P_l-P_{l-1}]$ to estimate $\E[P_L]$. 
Let $Y_l$ be the approximation of $\E[P_l-P_{l-1}]$
obtained by the quantum-accelerated Monte Carlo method (see \lem{Montanaro2015}).
By \lem{Montanaro2015} and \lem{powering lemma}, for any $\epsilon_l\geq 0$,
to make sure
$|\E[P_l-P_{l-1}] - Y_l| \leq \epsilon_l$ with probability at least $1-\delta$, we need $N_l=O(( 2^{-\hat{\beta} l}/\epsilon_l) (\log 2^{-\hat{\beta} l}/\epsilon_l)^{3/2}
(\log\log  2^{-\hat{\beta} l}/\epsilon_l)  (\log 1/\delta) )$ samples.
The error in approximating $\E[P_L]$ with 
$Y:=\sum_{l=0}^L Y_l$ satisfies
$|Y-\E[P_L]| \leq 
\sum_{l=0}^L |\E[P_l-P_{l-1}] - Y_l| \leq \sum_{l=0}^L \epsilon_l.
$
As a result, the error in approximating $\E[P]$ satisfies the bound
\be \label{eq:error}
|Y-\E[P]| \leq
|\E[P_L]-\E[P]| + |\E[P_L] - Y| \le |\E[P_L]-\E[P]| + 
\sum_{l=0}^L \epsilon_l
\ee
with probability at least $(1-\delta)^{L+1}$. 
The total cost equals 
$C = \sum_{l=0}^L C_lN_l$.

Choose 
\be \label{eq:choice of L}
L = \left\lceil \frac{\log(2 \epsilon^{-1})}{\alpha}
\right\rceil
\ee
so that 
$
2^{-\alpha L} \leq \epsilon/2.
$
This ensures the first error term
of \eq{error} is bounded by $\epsilon/2$.
As for the second term, we choose different $N_l$ based on the values $\epsilon_l,\hat{\beta},\gamma$.
Choosing $\delta = 1/(100(L+1))$, then we can make sure that the  success probability is at least
$(1-\delta)^{L+1} \geq e^{-1/100}>0.99$. 

We now split into cases to analyse and optimise the complexity of this approach. In the following analysis, for notational convenience we just write
$N_l= \lceil 2^{-\hat{\beta} l}/\epsilon_l\rceil
$. 
But in the end, the
cost should be multiplied by 
$ (\log 1/\delta)(\log 1/\epsilon)^{3/2} (\log\log 1/\epsilon) $ for all cases.

(a). If $\hat{\beta}>\gamma$, 
then choose $\epsilon_l= \frac{\epsilon}{2}(1-2^{-(\hat{\beta}-\gamma)/2} ) 2^{-(\hat{\beta}-\gamma)l/2}$. The second error term of \eq{error} is bounded by 
\be
\frac{\epsilon}{2}(1-2^{-(\hat{\beta}-\gamma)/2} )
\sum_{l=0}^L 2^{-(\hat{\beta}-\gamma)l/2} 
= \frac{\epsilon}{2}(1-2^{-(\hat{\beta}-\gamma)/2} )
\frac{1-2^{-(\hat{\beta}-\gamma)(L+1)/2}}{1-2^{-(\hat{\beta}-\gamma)/2}} 
< \frac{\epsilon}{2}.
\ee
Also we have
\be
N_l = 2
\epsilon^{-1} (1-2^{-(\hat{\beta}-\gamma)/2} )^{-1} 2^{-(\hat{\beta}+\gamma)l/2} + 1,
\ee
where the ``+1" term is caused by the ceiling function.
The cost is bounded by
\beas
\sum_{l=0}^L N_lC_l 
&=& O\left(\sum_{l=0}^L \left( \epsilon^{-1} (1-2^{-(\hat{\beta}-\gamma)/2} )^{-1} 2^{-(\hat{\beta}+\gamma)l/2} + 1\right) 2^{\gamma l} \right)\\
&=& O\left(\epsilon^{-1}(1-2^{-(\hat{\beta}-\gamma)/2} )^{-1} \sum_{l=0}^L 2^{-({\hat{\beta}-\gamma})l/2}
+ \sum_{l=0}^L2^{\gamma l} \right)   \\
&=& O\left( \epsilon^{-1}(1-2^{-(\hat{\beta}-\gamma)/2} )^{-1} 
\frac{1-2^{-(\hat{\beta}-\gamma)(L+1)/2} }{1-2^{-(\hat{\beta}-\gamma)/2}}
+ \epsilon^{-\gamma/\alpha} \right) \\
&=& O(\epsilon^{-1}),
\eeas
In the above, we used $\sum_{l=0}^L 2^{\gamma l} = O(2^{\gamma L}) = O(\epsilon^{-\gamma/\alpha})$ by our choice \eq{choice of L} and the fact that $\alpha\ge\gamma$.

(b). If $\hat{\beta}=\gamma$, then
set $\epsilon_l= \epsilon/(2(L+1))$.
This ensures that the error of the second term of
\eq{error} is bounded by $\epsilon/2$.
Thus $N_l= 2(L+1)\epsilon^{-1}  2^{-\hat{\beta} l} + 1 $. 
The cost is bounded by
\beas
\sum_{l=0}^L N_lC_l 
&=& O\left( \sum_{l=0}^L \left( (L+1) \epsilon^{-1}  2^{-\hat{\beta} l} + 1\right) 2^{\gamma l} \right)  \\
&=&  O\left((L+1)^2  \epsilon^{-1} 
+ \sum_{l=0}^L 2^{\gamma l}  \right)\\
&=& O\Bigl(\epsilon^{-1} (\log 1/\epsilon)^2 \Bigr) .
\eeas

(c). If $\hat{\beta}<\gamma$, then we set $\epsilon_l = \frac{\epsilon}{2} 2^{-(\gamma-\hat{\beta})L/2}
(1-2^{-(\gamma-\hat{\beta})/2}) 2^{(\gamma-\hat{\beta})l/2}$.
So the second error term of \eq{error} is bounded by 
\beas
\frac{\epsilon}{2} 2^{-(\gamma-\hat{\beta})L/2}
(1-2^{-(\gamma-\hat{\beta})/2}) \sum_{l=0}^L 2^{(\gamma-\hat{\beta})l/2}
&=& \frac{\epsilon}{2} 2^{-(\gamma-\hat{\beta})L/2}
(1-2^{-(\gamma-\hat{\beta})/2}) 
\frac{2^{(\gamma-\hat{\beta})(L+1)/2} - 1}{
2^{(\gamma-\hat{\beta})/2}-1} \\
&=& \frac{\epsilon}{2} (1-2^{-(\gamma-\hat{\beta})/2})
\frac{ 2^{(\gamma-\hat{\beta})/2} -  2^{-(\gamma-\hat{\beta})L/2}}{2^{(\gamma-\hat{\beta})/2}-1} \\
&<& \frac{\epsilon}{2} .
\eeas
Moreover, we have
\[
N_l = 2
\epsilon^{-1} 2^{(\gamma-\hat{\beta})L/2} (1-2^{-(\gamma-\hat{\beta})/2})^{-1} 2^{-(\hat{\beta}+\gamma)l/2} + 1
.
\]
The cost is bounded by
\beas
\sum_{l=0}^L N_lC_l 
&=& O\left(\sum_{l=0}^L ( \epsilon^{-1} 2^{(\gamma-\hat{\beta})L/2} (1-2^{-(\gamma-\hat{\beta})/2})^{-1} 2^{-(\hat{\beta}+\gamma)l/2} + 1) 2^{\gamma l} \right)\\
&=& O\left( \epsilon^{-1} 2^{(\gamma-\hat{\beta})L/2} (1-2^{-(\gamma-\hat{\beta})/2})^{-1} \sum_{l=0}^L 2^{(\gamma-\hat{\beta})l/2} +   \sum_{l=0}^L2^{\gamma l} \right) \\
&=& O\left( \epsilon^{-1} 2^{(\gamma-\hat{\beta})L/2} (1-2^{-(\gamma-\hat{\beta})/2})^{-1}
\frac{2^{(\gamma-\hat{\beta})(L+1)/2}-1}{2^{(\gamma-\hat{\beta})/2}-1} +  
\epsilon^{-\gamma/\alpha} \right)  \\
&=& O(\epsilon^{-1-(\gamma-\hat{\beta})/\alpha}) .
\eeas
In the above, we used equation \eq{choice of L} and
\[
2^{(\gamma-\hat{\beta})L} 
<2^{(\gamma-\hat{\beta})(\frac{\log(2 \epsilon^{-1})}{\alpha}+1)}
=2^{(\gamma-\hat{\beta})} 
2^{(\gamma-\hat{\beta})/\alpha}
\epsilon^{-(\gamma-\hat{\beta})/\alpha}.
\]

Since $L \approx \alpha^{-1} \log(2/\epsilon)$ and $\delta = 1/(100(L+1))$, we have
$\log(1/\delta)
=\log\log (2/\epsilon) + \log (100/\alpha) = O(\log\log 1/\epsilon)$.
Each estimation of the cost
should be multiplied by \\
$O((\log 1/\epsilon)^{3/2} (\log\log 1/\epsilon)^2)$.
\end{proof}

\section{Quantum-accelerated MLMC for solving SDEs}
\label{sec:SDE}

Let us discuss how to apply MLMC to solve \prb{problem} with stochastic differential equation~\eq{SDE}. 

\subsection{Preliminary}

Throughout the paper we make the following assumptions on the coefficients 
of the SDE and the payoff function. 
\begin{assumption}\label{ass:A1}
    We assume $\mu$ and $\sigma$ are globally Lipschitz continuous, \emph{i.e.}, 
    there exists a constant $L$ such that 
    \be \label{eq:A1}
    |\mu(t,x)-\mu(s,y)| \leq L(|t-s|+|x-y|),\quad  |\sigma(t,x)-\sigma(s,y)| \leq L(|t-s|+|x-y|)
    \ee
    hold for all $s, t \in [0,T], x,y \in \mathbb{R}$. We further assume the initial value $X_0$ satisfies $\E[X_0^m] \leq C_m$ for 
    constants $C_m \geq 0$. 
\end{assumption}
We remark that \ass{A1} implies at most linear growth of 
$\mu$ and $\sigma$, and there exists a unique strong solution of 
SDE~\eq{SDE}~\cite{KP13}. 

We say a numerical approximation $\widehat{X}_k$ with time step size $h=T/n$ is of strong order $r$, if for any $m \geq 1$, there exists a constant $C_m$ such that 
\begin{equation}\label{eq:def_strong_r}
    \mathbb{E}\left(\sup_{0\leq kh \leq T} |\widehat{X}_k-X_{kh}|^m\right) \leq C_mh^{rm}.
\end{equation}
One class of general high order schemes is the Taylor-It\^{o} scheme of the general form~\cite{KP13}
\begin{equation}
    \widehat{X}_{k+1} = \sum_{\alpha \in \mathcal{A}_{m}} f_{\alpha}(kh,\widehat{X}_k)I_{\alpha}
\end{equation}
where $f_{\alpha}$'s are the coefficient functions (depending on $\mu$ and $\sigma$) and $I_{\alpha}$ are 
multiple It\^{o} integrals over the time interval $[kh,(k+1)h]$. 
For instance, we may consider the Euler-Maruyama scheme (of strong order $1/2$)
\begin{equation}
\widehat{X}_{k+1} = \widehat{X}_k + \mu(\widehat{X}_k,t)h + \sigma(\widehat{X}_k,t) \Delta W_k,
\label{eq:Euler}
\end{equation}
for $k\in\rangez{n}$, or the Milstein scheme (of strong order $1$)
\begin{equation}
\widehat{X}_{k+1} = \widehat{X}_k + \mu(\widehat{X}_k,t)h + \sigma(\widehat{X}_k,t) \Delta W_k + \frac{1}{2}\sigma(\widehat{X}_k,t)\partial_X \sigma(\widehat{X}_k
,t)((\Delta W_k)^2-h),
\label{eq:Milstein}
\end{equation}
for $k\in\rangez{n}$, where $\Delta W_k$ are i.i.d. normal random variables with expected value zero and variance $h$.
We remark that there exists another kind of general high order schemes called Taylor-Stratonovich schemes ~\cite{KP13}, which is easier to implement. We will discuss it in \app{numerical results}.

\begin{assumption}\label{ass:A2}
    The coefficient functions $f_{\alpha}$ are globally Lipschitz 
    continuous with respect to $x$.
\end{assumption}

There exist a Taylor-It\^{o} scheme and a Taylor-Stratonovich scheme satisfying \ass{A2}, which can achieve strong order $r = k/2$ for all $k \geq 1$. We refer to~\cite[Section 10]{KP13} for more details. 

We are also given an assumption for the final payoffs:

\begin{assumption}\label{ass:A3}
We assume the payoff function
    \be \label{eq:P1}
    \P=\P(X_T)
    \ee
is piecewise Lipschitz continuous, \emph{i.e.}, there exist constants $-\infty = l_0<l_1<\cdots<l_q<l_{q+1}= +\infty$ and $L > 0$, such that 
    \be 
    |\P(x)-\P(y)| \leq L|x-y|, \qquad \forall x,y\in (l_j,l_{j+1}).
    \ee
\end{assumption}

In the reminder of this paper, we consider the SDEs, stochastic schemes, and payoffs satisfying \ass{A1}, \ass{A2}, \ass{A3}, respectively.
We remark that our results also hold true for high-dimensional systems of SDEs, given that the payoff function is ``piecewise Lipschitz continuous'' in some sense (\emph{e.g.} all the discontinuous points are jump discontinuous points and form several separable hyperplanes). 
For technical simplicity, we will only focus on the analysis of SDE in one dimension in this section. 

\subsection{Method and theory}

To solve a SDE problem, we apply the standard multilevel Monte Carlo method and regard the Taylor-It\^{o} scheme as the discretization subroutine \cite{Gil08, Gil15}.

At the high level, we estimate the discretized path $\widehat{X}_k$ ($k\in\rangez{n}$) for different numbers of iterations $n$, and perform the quantum oracle
\begin{equation}
U_P(|x\rangle|0\rangle) = |x\rangle|\P(x)\rangle
\label{eq:oracle}
\end{equation}
to evaluate $\P(x)$ for any $x$. 
Setting $n_l=2^l$ for $l=1,\ldots,L$, we apply the quantum-accelerated multilevel Monte Carlo
\begin{equation}
\E[\P(\widehat{X}_{n_L})] = \E[\P(\widehat{X}_{n_1})] + \sum_{l=1}^L[\E[\P(\widehat{X}_{n_l})-\P(\widehat{X}_{n_{l-1}})]]
\label{eq:multilevel}
\end{equation}
to estimate $\E[\P(X_T)]$. At the lower level, we divide $[0,T]$ by a uniform partition $0=t_0<t_1<\ldots<t_{n_l}=T$ with $h=T/n_l=T/2^l$ on the $l$-level discretization of \eq{SDE}, and perform stochastic numerical schemes to approximate $X_T$ by $\widehat{X}_{n_l}$. 

To estimate the complexity of QA-MLMC, we need to figure out the 
parameters $\alpha,\beta,\gamma$ in \thm{QMLMC}.

\begin{proposition}
\label{prop:alpha_beta_gamma_general}
Under \ass{A1}, \ass{A2} and \ass{A3}, for QA-MLMC with a numerical scheme of strong order $r$, we have $\alpha = r-o(1)$, $\beta = r-o(1)$, and $\gamma = 1$. Here $o(1)$ refers to an arbitrarily small real positive number. 
Furthermore, if the payoff function $\P$ is globally Lipschitz continuous, then the estimates on the parameters can be improved to 
$\alpha = r, \beta = 2r, \gamma = 1$. 
\end{proposition}
\begin{proof}
    The estimate on $\gamma$ comes from the construction of the algorithm. 
    At the level $l$, we use a time step size $T/2^l$ to discretize the path, and compute the expectation at the final time. 
    The dominant computational cost comes from simulating the path, which requires $2^l$ time steps. 
    At the level $l+1$, we halve the time step size, and the number of the time steps is doubled. 
    Since we are using the same numerical scheme at each level, the computational cost of propagating a single step remains the same, thus the total computational cost at the level $l+1$ is doubled. 
    This indicates that $\gamma = 1$. 
    We now focus on the estimate of $\alpha$ and $\beta$. 
    We first consider the general payoff function satisfying \ass{A3}. 
    
    The proof is inspired by \cite{giles2009analysing}. 
    For a sample of $X_T$ and a sample of numerical approximation 
    $\widehat{X}_n$ with time step size $h$, we define a linear path $\Lambda(\lambda) = \lambda \widehat{X}_n + (1-\lambda)X_T$ for $0 \leq \lambda \leq 1$. 
    Let $M(\widehat{X}_n,X_T) \in \{0,1,\cdots,q\}$ be the number of 
    the discontinuity points along the path. 
    Note that $\P$ is at least piecewise Lipschitz continuous, all the discontinuity points are of jump discontinuity. We hereby define the maximum size of the jump to be $J$, \emph{i.e.} 
    $$J = \max_{1\leq j\leq q}|\lim_{x\rightarrow l_j+}\P(x) - \lim_{x\rightarrow l_j-}\P(x)|.$$
    Then 
    \begin{align*}
        & \quad \mathbb{E}|\P(\widehat{X}_n)-\P(X_T)| \\
        &= \mathbb{E}\left[\mathbb{E}[|\P(\widehat{X}_n)-\P(X_T)|\Big|M(\widehat{X}_n,X_T)]\right]  \\
        &= \mathbb{E}[|\P(\widehat{X}_n)-\P(X_T)|\Big|M(\widehat{X}_n,X_T) = 0]\times \mathbb{P}(M(\widehat{X}_n,X_T) = 0) \\
        & \quad + \sum_{j=1}^q \mathbb{E}[|\P(\widehat{X}_n)-\P(X_T)|\Big|M(\widehat{X}_n,X_T) = j]\times \mathbb{P}(M(\widehat{X}_n,X_T) = j) \\
        & \leq L\mathbb{E}[|\widehat{X}_n-X_T|\Big|M(\widehat{X}_n,X_T) = 0]\times \mathbb{P}(M(\widehat{X}_n,X_T) = 0) \\
        & \quad + \sum_{j=1}^q\left[ qJ +  L\mathbb{E}[|\widehat{X}_n-X_T|\Big|M(\widehat{X}_n,X_T) = j]\right]\times \mathbb{P}(M(\widehat{X}_n,X_T) = j) \\
        & = L\mathbb{E}|\widehat{X}_n-X_T| + qJ\mathbb{P}(M(\widehat{X}_n,X_T) \geq 1). 
    \end{align*}
    The first part is bounded by $O(h^r)$. 
    The second part can be bounded as, for a large integer $m$, 
    \begin{align*}
        & \quad \mathbb{P}(M(\widehat{X}_n,X_T) \geq 1) \\
        &\leq \mathbb{P}(\min_{1\leq j\leq q} |X_T-l_j| \leq h^{\alpha}) + 
        \mathbb{P}(|\widehat{X}_n-X_T|\geq h^{\alpha}) \\
        & \leq O(h^{\alpha}) + \frac{\mathbb{E}|\widehat{X}_n-X_T|^m}{h^{\alpha m}} \\
        & \leq O(h^{\alpha}) + O(h^{m(r-\alpha)}).
    \end{align*}
    In the second inequality,
    the first term $O(h^{\alpha})$ 
    follows from that $X_T$ is a continuously-distributed random variable with a bounded density due to the
    Picard iteration used to establish existence and uniqueness \cite{mao2007stochastic} 
    under the global Lipschitz continuous assumption. 
    The second term follows from the Markov inequality.
    Therefore we have 
    \begin{equation}
        \mathbb{E}|\P(\widehat{X}_n)-\P(X_T)| \leq O(h^r) + O(h^{\alpha}) + O(h^{m(r-\alpha)})
    \end{equation}
    holds for all $m$, which implies $\alpha = rm/(m+1) = r - o(1)$. 
    
    The estimate of $\beta$ is similar to that of $\alpha$. 
    Note that $V_l = \V[\P_l-\P_{l-1}]$, so
    $V_l \leq (\sqrt{\V[\P-\P_l]} + \sqrt{\V[\P-\P_{l-1}]})^2$.
    It suffices to bound $\E[|\P-\P_l|^2]$, which is larger than $\V[\P-\P_l]$. 
    Using the same technique of estimating $\alpha$, 
    \begin{align*}
        & \quad \mathbb{E}|\P(\widehat{X}_n)-\P(X_T)|^2 \\
        &= \mathbb{E}[|\P(\widehat{X}_n)-\P(X_T)|^2\Big|M(\widehat{X}_n,X_T) = 0]\times \mathbb{P}(M(\widehat{X}_n,X_T) = 0) \\
        & \quad + \sum_{j=1}^q \mathbb{E}[|\P(\widehat{X}_n)-\P(X_T)|^2\Big|M(\widehat{X}_n,X_T) = j]\times \mathbb{P}(M(\widehat{X}_n,X_T) = j) \\
        & \leq L^2\mathbb{E}[|\widehat{X}_n-X_T|^2\Big|M(\widehat{X}_n,X_T) = 0]\times \mathbb{P}(M(\widehat{X}_n,X_T) = 0) \\
        & \quad + \sum_{j=1}^q\left[ 2q^2J^2 +  2L^2\mathbb{E}[|\widehat{X}_n-X_T|^2\Big|M(\widehat{X}_n,X_T) = j]\right]\times \mathbb{P}(M(\widehat{X}_n,X_T) = j) \\
        & \leq 2L^2\mathbb{E}|\widehat{X}_n-X_T|^2 + 2q^2J^2\mathbb{P}(M(\widehat{X}_n,X_T) \geq 1). 
    \end{align*}
    The first part is bounded by $O(h^{2r})$. 
    The second part is bounded by $O(h^{\beta})+O(h^{m(r-\beta)})$ for an arbitrarily large integer $m$, for the same reason in estimating $\alpha$. 
    It follows that 
    \begin{equation}
        \mathbb{E}|\P(\widehat{X}_n)-\P(X_T)|^2 \leq O(h^{2r}) + O(h^{\beta})+O(h^{m(r-\beta)}), 
    \end{equation}
    which implies $\beta = rm/(m+1) = r - o(1)$. 
    
    Finally, if further the payoff function $\P$ is globally Lipschitz 
    everywhere, then we have \\
    $\mathbb{P}(M(\widehat{X}_n,X_T) \geq 1) = 0$. 
    It is straightforward to conclude from the previous analysis 
    that $\alpha = r, \beta = 2r$. 
\end{proof}
    We remark that the estimates of $\alpha$ and $\beta$ are 
    possibly not sharp for some of the Taylor-It\^{o} schemes. 
    For example, if the payoff function $\P$ is a linear function, 
    then $\alpha$ will be exactly the weak convergence order, which is, 
    for many numerical schemes, larger than the strong convergence order. 
    Nevertheless, our \prop{alpha_beta_gamma_general} holds 
    true for more general payoff functions and general high order schemes, 
    and it suffices for QA-MLMC to achieve speedup over classical 
    algorithms. In \app{numerical results}, 
    we perform careful numerical tests of the values of $\alpha$ and $\beta$ under different smoothness assumptions. 
    We observe that our estimate for $\beta$ is sharp for both Lipschitz continuous payoff function and discontinuous payoff function, and our estimate for $\alpha$ is sharp for discontinuous payoff functions, while larger $\alpha$ is observed for smoother payoff functions.  
    
    \smallskip 
    
    \prop{alpha_beta_gamma_general} is  sufficient to determine the complexity of both classical and quantum-accelerated MLMC for solving SDEs. 
    We start with the classical case. A discussion about certain numerical schemes for typical payoffs has been proposed in Section 5 of \cite{Gil15}. For high-order schemes and general payoffs, we state the result as follows.
    
    \begin{proposition}\label{prop:MLMC_SDE}
        Consider \prb{problem} for the 
        stochastic differential equation \eq{SDE} under \ass{A1}, \ass{A2} and \ass{A3}. Then MLMC 
        with a numerical scheme for SDE of strong order $r$ estimates 
        $\mathbb{E}[\P]$ up to additive error $\epsilon$ with probability 
        at least 0.99 in cost 
         \be
         \begin{cases}
         O(\epsilon^{-2}), & r>1, \\
         O(\epsilon^{-1-1/r-o(1)}), & r\le 1.
         \end{cases}
         \ee
         Furthermore, if the payoff function $\P$ is globally Lipschitz 
         continuous everywhere, then the cost can be improved to 
         \be
         \begin{cases}
         O(\epsilon^{-2}), & r>1/2, \\
         O(\epsilon^{-2}(\log\epsilon)^2), & r=1/2, \\
         O(\epsilon^{-1/r}), & r<1/2.
         \end{cases}
         \ee
    \end{proposition}
    \begin{proof}
    This is a straightforward result from \lem{MLMC} and \prop{alpha_beta_gamma_general}. 
    \end{proof}
    
    \prop{MLMC_SDE} tells that, for general SDE \eq{SDE} and payoff function 
    satisfying \ass{A1}, \ass{A2}, and \ass{A3}, it suffices for MLMC 
    to use a numerical scheme of strong order $r>1$ 
    to obtain the complexity $O(\epsilon^{-2})$ , and using a numerical scheme of strong order $1$, \emph{e.g.}, Milstein scheme, will lead to the complexity $O(\epsilon^{-2-o(1)})$,
    \emph{i.e.} almost quadratic dependence on $1/\epsilon$. 
    We note that if the payoff function is globally Lipschitz continuous 
    everywhere, 
    then it suffices to use a numerical scheme of strong order $1/2$, \emph{e.g.}, Euler-Maruyama scheme, to achieve the complexity $\widetilde{O}(\epsilon^{-2})$.
    
    As a counterpart, we are ready to state our main theorem regarding the complexity 
    of QA-MLMC for solving SDE. 
    
    \begin{theorem}\label{thm:QMLMC_SDE}
        Consider \prb{problem} for the 
        stochastic differential equation \eq{SDE} under \ass{A1}, \ass{A2} and \ass{A3}. Then QA-MLMC 
        with a numerical scheme for SDE of strong order $r$ estimates 
        $\mathbb{E}[\P]$ up to additive error $\epsilon$ with probability 
        at least 0.99 in cost 
        \be
        \begin{cases} \vspace{.1cm}
          O\Bigl(\epsilon^{-1} (\log 1/\epsilon)^{3/2} (\log\log 1/\epsilon)^2\Bigr), & r  > 2, \\
          O\Bigl(\epsilon^{-1/2-1/r-o(1)}(\log 1/\epsilon)^{3/2} (\log\log 1/\epsilon)^2\Bigr), & r \leq 2.
         \end{cases}
         \ee
         Furthermore, if the payoff function $\P$ is globally Lipschitz 
         continuous everywhere, then the cost can be improved to 
         \be
         \begin{cases} \vspace{.1cm}
          O\Bigl(\epsilon^{-1} (\log 1/\epsilon)^{3/2} (\log\log 1/\epsilon)^2\Bigr), & r > 1, \\ \vspace{.1cm}
          O\Bigl(\epsilon^{-1}(\log 1/\epsilon)^{7/2} (\log\log 1/\epsilon)^2\Bigr)
          , & r = 1, \\
          O\Bigl(\epsilon^{-1/r} (\log 1/\epsilon)^{3/2} (\log\log 1/\epsilon)^2\Bigr), & r < 1.
          \end{cases}
          \ee
    \end{theorem}
    \begin{proof}
    This is a straightforward result from \thm{QMLMC} and \prop{alpha_beta_gamma_general}. 
    \end{proof}
    
    \thm{QMLMC_SDE} tells that, for general SDE \eq{SDE} and payoff function 
    satisfying \ass{A1}, \ass{A2}, and \ass{A3}, it suffices for QA-MLMC 
    to use a numerical scheme of strong order $r>2$ 
    to obtain the complexity $\widetilde{O}(\epsilon^{-1})$, and using a numerical scheme of strong order $2$ will lead to the complexity $\widetilde{O}(\epsilon^{-1-o(1)})$, \emph{i.e.} almost linear dependence on $1/\epsilon$. 
    We note that if the payoff function is globally Lipschitz continuous 
    everywhere, 
    then it suffices to use a numerical scheme of strong order $1$, \emph{e.g.}, Milstein scheme, to achieve the complexity $\widetilde{O}(\epsilon^{-1})$. 
    
    Compared with the classical MLMC, the convergence order required to achieve possibly optimal complexity is higher. This is due to the tighter requirement on the parameters $\alpha,\beta$ and $\gamma$. Nevertheless, once the optimal complexity is reached, QA-MLMC achieves a quadratic speedup over the classical MLMC in terms of $\epsilon$. 
    
\subsection{Generalization to the entire path dependence case}    

So far we have confined ourselves to the case that the payoff function $\P$ only depends on the final value $X_T$ of the stochastic process. 
However, the payoff functions of some options of widely practical interest, including Asian call and put options, are in general dependent on the stochastic integral along the entire trajectory of $X_t$ for any $0 \leq t \leq T$. 
Furthermore, as being discussed later in \sec{Greeks}, one of the most efficient ways to estimate the sensitivity of the portfolio involves  calculating the expectation of a stochastic integral. 
Fortunately, QA-MLMC can be straightforwardly generalized to this situation, and we will elaborate it in this subsection. 

We first allow the payoff function $\P$ to be possibly dependent on the stochastic integral of the entire trajectory through the following assumption: 
\begin{assumption}\label{ass:A4}
    We assume the payoff function has the form 
    \begin{equation}\label{eq:P2}
    \P = \P\left(X_T, \int_0^T f(X_t)\d t, \int_0^T g(X_t)\d W_t\right),
    \end{equation}
    where $\P(X,y,z)$ is piecewise Lipschitz continuous with respect to $X,y,z$, and $f$ and $g$ are two globally Lipschitz continuous functions. 
\end{assumption}
Such a payoff function can be transformed to only depend on the final value of another stochastic process by extending the system to higher  dimension. 
Precisely, let us define $$Y_t = \int_0^t f(X_s)ds, \quad Z_t = \int_0^t g(X_s)\d W_s, $$
then $Y_t$ and $Z_t$, together with $X_t$, satisfy the system of stochastic differential equations 
\begin{equation}
    \begin{split}
        dX_t &= \mu(X_t,t)\d t + \sigma(X_t,t)\d W_t, \\
        dY_t &= f(X_t)\d t, \\
        dZ_t &= g(X_t)\d W_t. 
    \end{split}
\end{equation}
This can be formally written as a system of stochastic differential equations in higher dimension, \emph{i.e. } $$d\widetilde{X}_t = \widetilde{\mu}(\widetilde{X}_t,t)\d t + \widetilde{\sigma}(\widetilde{X}_t,t)d\widetilde{W}_t,$$ 
where $\widetilde{X}_t^{\top} = (X_t^{\top},Y_t^{\top},Z_t^{\top})$, $\widetilde{W}_t$ a Brownian motion in the dimension of $\widetilde{X}_t$, 
\begin{equation}
    \widetilde{\mu} = \left(\begin{array}{c}
         \mu(X_t,t) \\
         f(X_t) \\
         0 
    \end{array}\right), \quad 
    \widetilde{\sigma} = \left(\begin{array}{ccc}
        \sigma(X_t,t) & 0 & 0 \\
        0 & 0 & 0 \\
        g(X_t) & 0 & 0
    \end{array}\right).
\end{equation}
Define 
$$\widetilde{\P}(\widetilde{X}_t) =  \P\left(X_t, Y_t, Z_t\right),$$ 
then 
$$\widetilde{\P}(\widetilde{X}_T) =  \P\left(X_T, Y_T, Z_T\right) = \P\left(X_T, \int_0^T f(X_t)\d t, \int_0^T g(X_t)\d W_t\right), $$ 
and thus estimating the expectation of $\P$ is equivalent to solving \prb{problem} in terms of $\widetilde{X}_t$ and $\widetilde{\P}$. 
By \ass{A4}, we can check that $\widetilde{X}_t$ and $\widetilde{\P}$ satisfy the \ass{A1}, \ass{A2} and \ass{A3}. 
Therefore, according to \prop{MLMC_SDE} and \thm{QMLMC_SDE}, we have the following result. 
\begin{corollary}\label{cor:MLMC_SDE_integral}
        Consider the payoff function satisfying \ass{A4}. Then, further under \ass{A1} and \ass{A2}, MLMC 
        with a numerical scheme for SDE of strong order $r$ estimates 
        $\mathbb{E}[\P]$ up to additive error $\epsilon$ with probability 
        at least 0.99 in cost 
         \be
         \begin{cases}
         O(\epsilon^{-2}), & r>1, \\
         O(\epsilon^{-1-1/r-o(1)}), & r\le 1.
         \end{cases}
         \ee
         Furthermore, if the payoff function $\P$ is globally Lipschitz 
         continuous everywhere, then the cost can be improved to 
         \be
         \begin{cases}
         O(\epsilon^{-2}), & r>1/2, \\
         O(\epsilon^{-2}(\log\epsilon)^2), & r=1/2, \\
         O(\epsilon^{-1/r}), & r<1/2.
         \end{cases}
          \ee
\end{corollary}

\begin{corollary}\label{cor:QMLMC_SDE_integral}
        Consider the payoff function satisfying \ass{A4}. Then, further under \ass{A1} and \ass{A2}, QA-MLMC 
        with a numerical scheme for SDE of strong order $r$ estimates 
        $\mathbb{E}[\P]$ up to additive error $\epsilon$ with probability 
        at least 0.99 in cost 
        \be
        \begin{cases} \vspace{.1cm}
          O\Bigl(\epsilon^{-1} (\log 1/\epsilon)^{3/2} (\log\log 1/\epsilon)^2\Bigr), & r  > 2, \\
          O\Bigl(\epsilon^{-1/2-1/r-o(1)}(\log 1/\epsilon)^{3/2} (\log\log 1/\epsilon)^2\Bigr), & r \leq 2.
         \end{cases}
         \ee
         Furthermore, if the payoff function $\P$ is globally Lipschitz 
         continuous everywhere, then the cost can be improved to 
         \be
         \begin{cases} \vspace{.1cm}
          O\Bigl(\epsilon^{-1} (\log 1/\epsilon)^{3/2} (\log\log 1/\epsilon)^2\Bigr), & r > 1, \\ \vspace{.1cm}
          O\Bigl(\epsilon^{-1}(\log 1/\epsilon)^{7/2} (\log\log 1/\epsilon)^2\Bigr)
          , & r = 1, \\
          O\Bigl(\epsilon^{-1/r} (\log 1/\epsilon)^{3/2} (\log\log 1/\epsilon)^2\Bigr), & r < 1.
          \end{cases}
          \ee
\end{corollary}

We next discuss several applications arising in mathematical finance.

\section{Black-Scholes option pricing model}
\label{sec:BSM}

We consider the Black-Scholes model for option pricing \cite{Bod09,BS73,Hul03}, which is at the core of quantitative finance, as the first application of \prb{problem}.

\subsection{Black-Scholes equation}
The Black-Scholes model is used to price a variety of financial derivatives via a simple and analytical solvable model using a small number of input parameters. In Black-Scholes model, we are mainly interested in the following Geometric Brownian Motion
\begin{equation}
\d{S_t} = \mu S_t\d t + \sigma S_t\d W_t,
\label{eq:BSM}
\end{equation} 
where $S_t$ is the asset price, $\sigma$ the velocity of the asset and $\mu$ is the risk-free interest rate. The rate of return on the asset $\mu$ is assumed to be a constant. There are some other assumptions on the model, which Fischer Black and Myron Scholes have thoroughly discussed in \cite{BS73}.  For path-independent option, the price of option $V(s,t)$ follows the Black-Scholes equation
\begin{equation}
\frac{\partial V}{\partial t} + \frac{\sigma^2s^2}{2}\frac{\partial^2 V}{\partial s^2} + \mu s\frac{\partial V}{\partial s} = \mu V
\label{eq:BSE}
\end{equation}
for $s>0$, $t\le T$, with a terminal condition $V(s,T)=\psi(s)$, where $\psi$ is the final payoff. The link between \eq{BSE} and \eq{BSM} is established by the Feyman-Kac formula. The solution of \eq{BSE} can be represented as the expectation of the solution of \eq{BSM}.

\begin{lemma}\label{lem:BSE}
Consider a PDE defined in \eq{BSE} subject to a terminal condition $V(s,T)=\psi(s)$. The solution $V(s,t)$ can be written as a conditional expectation
\begin{equation}
V(s,t) = \E[ e^{-\mu (T-t)}\psi(S_T) ~|~ S_t = s],
\label{eq:FK_BSE}
\end{equation}
where $S_t$ is an It\^{o} process driven by \eq{BSM}.
\end{lemma}

Using It\^{o} lemma, \eq{BSM} under the condition $S_t=s$ can be solved as 
\begin{equation}
    S_T=s e^{\sigma W_{T-t} +(\mu -\sigma^2/2)(T-t)}.
\label{eq:BSMsol}
\end{equation}
For European options, one can analytically solve for $V(s,t)$, which has a deterministic expression. However, in the case of more complex payoff functions, one needs to resort to the Monte Carlo method, which gives an approximation of $V(s,t)$ by averaging the payoff over samples $W_t$. The quantum version of the Monte Carlo method can be applied to reduce the complexity from $\widetilde O (\epsilon^{-2})$ to $\widetilde O( \epsilon^{-1})$, and the examples of Black-Scholes model for European options and Asian options are thoroughly discussed in \cite{RGB18}.

Nevertheless, the analytical solution \eq{BSMsol} can be used for benchmarking numerical simulations to understand the scaling behavior of discretization schemes. In \app{numerical results}, we present some numerical tests on the Black-Scholes equation to demonstrate our theory results. In our method, we apply discretization schemes to \eq{BSM} to simulate random paths of $S_t$, such as Euler-Maruyama scheme
\begin{equation}
\widehat{S}_{k+1} = \widehat{S}_k + \mu \widehat{S}_kh + \sigma \widehat{S}_k\Delta W_k,
\label{eq:Euler_BS}
\end{equation}
and Milstein scheme
\begin{equation}
\widehat{S}_{k+1} = \widehat{S}_k + \mu \widehat{S}_kh + \sigma \widehat{S}_k\Delta W_k + \frac{\sigma^2}{2} \widehat{S}_k((\Delta W_k)^2-h).
\label{eq:Milstein_BS}
\end{equation}
Here, \eq{Euler_BS} and \eq{Milstein_BS} are reduced from \eq{Euler} and \eq{Milstein} respectively. Then we estimate $\E[P(S_T)]$ with $P(S_t)=e^{-\mu (T-t)}\psi(S_T)$, where the quantum oracle \eq{oracle} can be reduced to
\begin{equation}
U_{\psi}(|x\rangle|0\rangle) = |x\rangle|\psi(x)\rangle.
\label{eq:oracle_BS}
\end{equation}
It is worthwhile to mention that our method is suitable for general SDEs and payoff functions. 

\subsection{Option pricing}

Black-Scholes model provides a useful tool for option pricing, which was a longstanding problem in finance. Briefly speaking, an option is a contract that allows the holder to buy or sell a financial asset at a fixed price in the future. A call option is an option to buy an asset and a put option is an option to sell it. In this subsection, we would like to present some well-known call options as examples. Put options can be developed in a similar way.

We first consider Lipschitz continuous options. One of the famous options is \textit{European option}, in which the final payoff is given by
\begin{equation}
\psi(S_T) = (S_T-K)^+ := \max\{S_T-K , 0\},
\label{eq:European}
\end{equation}
where $K>0$ is the strike price of the option. That means that if $S_T \leq K$, the option is worthless, and if $S_T>K$, the holder can buy the asset for $K$ dollars and sell it at market price, making a profit of $S_T-K$. Note that the European option is path independent and only relies on the terminal price $S_T$, without the consideration of the whole path.  

There also exist path-dependent options relying on the path $\{S_t\}$. One example is \textit{Asian option}, in which the final payoff is considered to be
\begin{equation}
\psi(S_T) = (\frac{1}{T}\int_0^TS_t\d{t}-K)^+
\label{eq:Asian}
\end{equation}
with the strike $K>0$. The payoff of Asian option is determined by the average of the asset price over $[0,T]$.



It follows from a theorem that the classical multilevel Monte Carlo method with Euler-Maruyama scheme achieves the complexity $\widetilde O(\epsilon^{-2})$ for globally Lipschitz continuous options, which improves the complexity of standard Monte Carlo method $\widetilde O(\epsilon^{-3})$. \cite{Gil08, Gil15}.

We then consider piecewise Lipschitz continuous options. A typical option is \textit{Digital option}, in which the digital option (Cash-or-nothing option) gives the final payoff as the form
\begin{equation}
\psi(S_T) = \mathcal{H}(S_T-K),
\label{eq:digital}
\end{equation}
with the strike $K>0$, where $\mathcal{H}$ is the Heaviside function. Clearly it is a non-Lipschitz continuous option.

For piecewise Lipschitz continuous options, classical multilevel Monte Carlo method with Euler-Maruyama scheme achieves the complexity $\widetilde O(\epsilon^{-2.5})$, and it can be further improved to $\widetilde O(\epsilon^{-2})$ by Milstein scheme and extreme paths~\cite{Gil15}.

By \thm{QMLMC_SDE} and \cor{QMLMC_SDE_integral}, the quantum-accelerated multilevel Monte Carlo method can be used to obtain the following result:

\begin{corollary}\label{cor:piecewise_BS}
We consider \prb{problem} given by a Geometric Brownian Motion \eq{BSM} with an initial condition $S_t=s_0$. We are given the zeroth-order quantum oracle \eq{oracle_BS} for a final piecewise Lipschitz continuous payoff $\P$ as defined in \eq{P1} or \eq{P2}. Then QA-MLMC with a numerical scheme for SDE of strong order $r$ estimates $\mathbb{E}[P]$ up to additive error $\epsilon$ with probability at least 0.99 in cost 
\be
\begin{cases}
\widetilde O(\epsilon^{-1}), & r  > 2, \\
\widetilde O(\epsilon^{-1/2-1/r-o(1)}), & r \leq 2.
\end{cases}
\ee
Furthermore, if the payoff function $\P$ is globally Lipschitz continuous everywhere, then the cost can be improved to 
\be
\begin{cases}
\widetilde O(\epsilon^{-1}), & r  \geq 1, \\
\widetilde O(\epsilon^{-1/r}), & r < 1.
\end{cases}
\ee
\end{corollary}

For globally Lipschitz continuous $P$, it suffices to perform the Milstein scheme \eq{Milstein_BS} with strong order $r=1$ to achieve the complexity $\tilde{O}(\epsilon^{-1})$. While for piecewise Lipschitz continuous $P$, it requires to use a numerical scheme of strong order $5/2$ to obtain the complexity $\tilde{O}(\epsilon^{-1})$.

\section{Local Volatility model}
\label{sec:LVM}

The Local Volatility model generalizes the Black-Scholes model \eq{BSM} by treating volatility $\sigma$ as a function of the asset $S_t$ and the time $t$ \cite{DK94,Dup94}. The Local Volatility model is a kind of simplification of the stochastic volatility model, which assumes $\sigma$ has a randomness of its own. Differing from the analytically solvable Black-Scholes model, we are required to simulate the SDE, as there is a lack of explicit solutions to estimate the price of the Local Volatility model.

The Local Volatility model characterizes $S_t$ by a generalized Geometric Brownian Motion as the form
\begin{equation}
\d{S_t} = \mu S_t\d t + \sigma(S_t,t) S_t\d W_t
\label{eq:LVM}
\end{equation}
for $t\le T$, with an initial condition $S_t=s_0$. Here $\mu$ is the instantaneous risk-free interest rate, $\sigma(S_t,t)$ is the instantaneous volatility of the risky asset, which is sufficiently smooth with respect to $S_t$ and $t$, and $W_t$ is a standard Brownian motion. We further require \eq{LVM} satisfies the assumptions in \sec{SDE}.

By Feynman-Kac Formula, the SDE \eq{LVM} corresponds to the Black-Scholes equation that describes the price of the option $V=V(s,t)$ by
\begin{equation}
\frac{\partial V}{\partial t} + \frac{\sigma^2(s,t)s^2}{2}\frac{\partial^2 V}{\partial s^2} + \mu s\frac{\partial V}{\partial s} = \mu V
\label{eq:LVE}
\end{equation}
for $s>0$, $t\le T$, with a terminal condition $V(s,T)=\psi(s)$, where $\psi$ is the final payoff. 

We establish  the  link  between \eq{LVM} and \eq{LVE} by Feynman-Kac Formula.

\begin{lemma}\label{lem:LVE}
Consider a PDE defined in \eq{LVE} subject to a terminal condition $V(s,T)=\psi(s)$. The solution $V(s,t)$ can be written as a conditional expectation
\begin{equation}
V(s,t) = \E[ e^{-\mu (T-t)}\psi(S_T) ~|~ S_t = s ],
\label{eq:FK_LVE}
\end{equation}
where $S_t$ is an It\^{o} process driven by \eq{LVM}.
\end{lemma}

Thus, we can estimate $\E[P(S_T)]$ with $P(S_t)=e^{-\mu (T-t)}\psi(S_T)$ to obtain $V(s,t)$ for general options.

We consider the same setting for payoffs in \sec{BSM}. By \thm{QMLMC_SDE} and \cor{QMLMC_SDE_integral}, the quantum-accelerated multilevel Monte Carlo method can be used to obtain the following results:

\begin{corollary}\label{cor:piecewise_LV}
We consider \prb{problem} given by a SDE \eq{LVM} with an initial condition $S_t=s_0$. We are given the zeroth-order quantum oracle \eq{oracle_BS} for a final piecewise Lipschitz continuous payoff $\P$ as defined in \eq{P1} or \eq{P2}. Then QA-MLMC with a numerical scheme for SDE of strong order $r$ estimates $\mathbb{E}[P]$ up to additive error $\epsilon$ with probability at least 0.99 in cost 
\be
\begin{cases}
\widetilde O(\epsilon^{-1}), & r  > 2, \\
\widetilde O(\epsilon^{-1/2-1/r-o(1)}), & r \leq 2.
\end{cases}
\ee
Furthermore, if the payoff function $\P$ is globally Lipschitz continuous everywhere, then the cost can be improved to 
\be
\begin{cases}
\widetilde O(\epsilon^{-1}), & r  \geq 1, \\
\widetilde O(\epsilon^{-1/r}), & r < 1.
\end{cases}
\ee
\end{corollary}

\section{Sensitivity of price}
\label{sec:Greeks}

Besides the price, another important factor of financial products is the risk. If two financial products have similar expected return, the one with lower risk is usually preferable because people would like to ensure a positive rate of return earnings under most situations and are not willing to take a risk. One common way to reduce the risk is hedging, that is, forming a portfolio by buying several financial products with different reaction and sensitivity to the change of the market. For example, if there are two stocks, and one is positively correlated with the global economy and the other one is negatively correlated, then one may purchase both of them with a specific ratio to guarantee the return no matter whether the global economy flourishes or declines. Correctly identifying and measuring the sensitivity is crucial to hedging portfolios. 

``The Greeks" label the sensitivity of the price of the option \cite{Hul03}. They are partial derivatives with respect to the parameters such as the initial price, the starting time, the risk-free rate and the volatility. Computing these is essential to hedge portfolios, and therefore it is even more important than pricing the option itself. Different from the Black-Scholes model, there's no closed form for Greeks with general payoff functions and coefficients in SDEs. Let $u(s,t)$ denote $\E(\P(X_T))$ where $X_T$ is the solution of SDE
\begin{equation}
\d{X_\tau} = \mu(X_\tau)\d \tau + \sigma(X_\tau)\d W_\tau, \quad \tau \in [t,T]
\end{equation}
with starting time $t$ and initial condition $X_t = s$. 
Different types of Greeks are as follow: 
\begin{itemize}
    \item Delta:
\begin{equation}
\frac{\partial u(s,t)}{\partial s},
\end{equation}
    \item Gamma:
\begin{equation}
\frac{\partial^2 u(s,t)}{\partial s^2},
\end{equation}
    \item Vega:
\begin{equation}
\frac{\partial u(s,t)}{\partial \sigma},
\end{equation}
    \item Theta:
\begin{equation}
\frac{\partial u(s,t)}{\partial t},
\end{equation}
    \item Rho:
\begin{equation}
\frac{\partial u(s,t)}{\partial r}.
\end{equation}
\end{itemize}

Numerical treatments of Greeks have been widely studied in \emph{e.g. }\cite{Fournie1999Greeks, Malliavin2006Greeks, Paskov1995Greeks, Liu2010Greeks, Giles2012Greeks}. 
A natural way to compute the Greeks is finite difference, that is, to compute the expectation of the payoff function with different parameters and use finite difference to approximate the derivatives. 
Although finite difference approximation of derivatives suffers from numerical instability, using common stochastic paths for both estimators in finite difference can reduce the classical complexity. 
However, if the payoff function is not smooth enough, finite difference might not be effective. 
Here we consider an alternative approach using Malliavin calculus~\cite{Fournie1999Greeks} to compute the Greeks. 
Although rigorous derivation of using Malliavin calculus to compute Greeks is technically involved, the basic idea is simply using an analog of calculus of variations 
and integration by parts formula in the stochastic calculus setting, and the outcome formula to compute Greeks is amazingly concise thus easy to implement. 
We will show how to combine it with QA-MLMC to achieve quadratic speedup. 

For expository purpose let us consider the example of computing Delta in one dimension and omit a few technical details. 
We refer interested readers to~\cite{Fournie1999Greeks} for elaborations and other types of the Greeks. 
Assume the process $X_t$ is given by \eq{SDE}, and for simplicity we assume that the system is autonomous, \emph{i.e. }there is no explict time dependence in $\mu$ and $\sigma$. 
Define another stochastic process $Y_t$ to be 
\begin{equation}\label{eq:Greeks_SDE_Y}
    dY_t = \mu'(X_t)Y_t\d t + \sigma'(X_t)Y_t\d W_t, \quad Y_0 = 1. 
\end{equation}
Under \ass{A1}, \ass{A2}, \ass{A3}, and several further technical assumptions that $\mu,\sigma$ are $C^1$ functions, $\sigma$ is uniformly bounded away from 0 and the payoff function $\P$ has uniformly bounded second moment,~\cite[Proposition 3.2]{Fournie1999Greeks} tells that Delta can be represented via the following formula: 
\begin{equation}\label{eq:Greeks_exp}
    \frac{\partial u(s,0)}{\partial s} = \E \left[\frac{1}{T}\P(X_T)\int_0^TY_t\sigma^{-1}(X_t)\d W_t ~|~ X_0=s,Y_0=1\right]. 
\end{equation}
Notice that here we only assume the payoff function $\P$ to be piecewise Lipschitz continuous. 
Therefore, we can compute Delta by sampling $(X_t,Y_t)$ according to \eq{SDE} and \eq{Greeks_SDE_Y}, and then use QA-MLMC to compute the expectation \eq{Greeks_exp}, which is in the form of \ass{A4}. 
By \cor{QMLMC_SDE_integral}, we have the following complexity estimate: 
\begin{corollary}\label{cor:Greeks}
Assume we are given the zeroth-order quantum oracle \eq{oracle_BS} for a final piecewise Lipschitz continuous payoff $\P$ as defined in \eq{P2}. Then by sampling $(X_t,Y_t)$ according to \eq{SDE} and \eq{Greeks_SDE_Y} with a numerical scheme of strong order $r$ and estimating Delta via \eq{Greeks_exp} using QA-MLMC, the Delta can be approximated up to an additive error $\epsilon$ with probability at least 0.99 in cost 
\be
\begin{cases}
\widetilde O(\epsilon^{-1}), & r  > 2, \\
\widetilde O(\epsilon^{-1/2-1/r-o(1)}), & r \leq 2.
\end{cases}
\ee
Furthermore, if the payoff function $\P$ is globally Lipschitz continuous, then the cost can be improved to 
\be
\begin{cases}
\widetilde O(\epsilon^{-1}), & r  \geq 1, \\
\widetilde O(\epsilon^{-1/r}), & r < 1.
\end{cases}
\ee
\end{corollary}

\section{Binomial option pricing model}
\label{sec:BLM}

The binomial option pricing model, a.k.a. the binomial lattice model,  provides a discrete time (lattice-based) approximation to the Black–Scholes model \cite{CRR79,Hul03,Ren79}. The binomial model assumes that movements in the price follow a binomial distribution, which approaches the log-norm distribution of the Geometric Brownian Motion.

Given a risk-neutral measure and an initial condition $S_0$, we divide $[0,T]$ by a uniform partition $0=t_0<t_1<\ldots<t_n=T$ with $h=T/n$, and model the price of a stock in discrete time by a Markov chain
\begin{equation}
S_{k+1} = S_kY_{k+1},
\label{eq:BLM}
\end{equation}
where $\{Y_k\}$ $k\in\rangez{n}$ are i.i.d. with a two points ``up'' $U$ and ``down'' $D$ distribution
\begin{equation}
\begin{aligned}
& \mathrm{Pr}(Y_k = U) = p;\\
& \mathrm{Pr}(Y_k = D) = 1-p.\\
\end{aligned}
\end{equation}
Under this setting, the probability of the value $S_N$ is given by
\begin{equation}
\mathrm{Pr}(S_n = U^kD^{n-k}S_0) = \Bigl(^n_k\Bigr) p^k(1-p)^{n-k}.
\end{equation}
Thus, we create a binomial tree with distribution $B(n,p)$ that describes the prices over time.

We require the conditional expectation of \eq{BLM} matches the Geometric Brownian Motion \eq{BSM}, giving
\begin{equation}
pS_kU+(1-p)S_kD = \E[S_{k+1}~|~S_k] = S_ke^{rh},
\label{eq:BLM_mean}
\end{equation}
which is
\begin{equation}
p = \frac{e^{rh}-D}{U-D}.
\end{equation}
Practially, we could use $\E[S_{k+1}~|~S_k] \approx S_k(1+rh)$ instead, which in fact corresponds with the expectation of the Euler-Maruyama scheme \eq{Euler_BS}. Similarly, we require the conditional variance of \eq{BLM} matches the geometry Brownian motion \eq{BSM},
\begin{equation}
pS_k^2U^2+(1-p)S_k^2D^2-[pS_k^2U+(1-p)S_k^2D]^2 = \mathrm{Var}(S_{k+1}~|~S_k) = S_k^2e^{2rh}(e^{\sigma^2h}-1) \approx S_k^2\sigma^2h.
\label{eq:BLM_var}
\end{equation}

There are various binomial lattice models satisfying \eq{BLM_mean} and \eq{BLM_var}. The first and the most famous model is the CRR model \cite{CRR79}. If we choose step size $h \le \sigma^2/r^2$ , and assume $D = 1/U$, we obtain
\begin{equation}
U = e^{\sigma\sqrt{h}}, \qquad D = e^{-\sigma\sqrt{h}}.
\label{eq:CCR}
\end{equation}
Alternatively, we can also consider an equal probabilities model, the JR model \cite{JR83}. If we choose $p=1/2$ to determine $U$ and $D$, we would obtain
\begin{equation}
U = e^{(r-\sigma^2/2)h+\sigma\sqrt{h}}, \qquad D = e^{(r-\sigma^2/2)h-\sigma\sqrt{h}}.
\label{eq:JR}
\end{equation}
Similar as before, we can use first-order approximations of exponential factors as values of $U$ and $D$ in practice.

Now we perform BOPM to estimate expectation of option pricing $\E[P(S_T)]$ The error of approximating general Black-Scholes option prices by CCR model \eq{CCR} is $O(1/\sqrt{n})$ \cite{HZ00}, while it can be further improved to $O(1/n)$ for European option \cite{LR96}. 

In general, BOPM works for various options with complexity $O(n)$, since we need to calculate $S_n = U^kD^{n-k}S_0$. Even if $U=1/D$, we should still multiphy $U$ in total $2k-n$ times. In general, we can apply multilevel MC to reduce the cost of BOPM. Since we are required to choose $n=O(1/\epsilon^2)$, its complexity is generally no better than Euler-Maruyama scheme.

But for the piecewise constant payoff, such as the digital option \eq{digital}, it allows us to develop classical and quantum algorithms that only require $O(\log N)$ time complexity. Given $\overline S_0 < \cdots < \overline S_m$, we define a piecewise constant payoff by
\begin{equation}
\psi(S_n) =
\begin{cases}
\psi_L, & S_n < \overline S_0, \\
\psi_j, & \overline S_{j-1} \le S_n < \overline S_j, \\
\psi_R, & S_n \ge \overline S_m,
\end{cases}
\label{eq:constant}
\end{equation}
where $\psi_j$, $\psi_L$, $\psi_R$ are constant. Noting that \eq{constant} is a special case of \ass{A3}. We are also given a zeroth-order classical oracle
\begin{equation}
O_{\psi}(S_0,k) = \psi(U^kD^{n-k}S_0),
\label{eq:constant_classical_oracle}
\end{equation}
which outputs the final payoff \eq{constant}. The procedure is described as follows: given $S_0$, $U$, $D$, $N$, we can determine the criteria $k^{\ast}_j$ corresponding with $\overline S_j$, by
\begin{equation}
k^{\ast}_j\log U + (n-k^{\ast}_j)\log D + \log S_0 = \log \overline S_j. 
\end{equation}
When we input $n$, we just compare $n$ with $n^{\ast}_j$, and determine the payoff \eq{constant}. Under this setting, the complexity of \eq{constant_classical_oracle} is independent on $n$, by avoiding $n$ times multiplication for calculating $S_n=U^kD^{n-k}S_0$.

We follow Theorem 12 of \cite{LMS20}, in which we replace the fast random walk by BOPM.

\begin{proposition}\label{prop:BOPM}
We consider \prb{problem} given by a Geometric Brownian Motion \eq{BSM} with an initial condition $S_0$. We are given the zeroth-order classical oracle \eq{constant_classical_oracle} for a final piecewise constant payoff as defined in \eq{constant}, which has a bounded variance independent of $h$. There exists a classical algorithm that estimates $\E[P]$up to additive error $\epsilon$ with probability at least 0.99 in cost 
\begin{equation}
\widetilde O(\epsilon^{-2}).
\end{equation}
\end{proposition}

\begin{proof}
We aim to perform $n = O(1/\epsilon^2)$ steps BOPM \eq{BLM} to obtain $S_T$ within error $\epsilon$, and then calculate $P(S_T)$ as one sample of Monte Carlo simulation. Given $U$, $D$, $S_0$ sampled from $\pi_0$ in $O(1)$, and the oracle \eq{constant_classical_oracle} for \eq{constant}, we do the sampling procedure as follows.

Inspired by Lemma 11 and Theorem 12 of \cite{LMS20}, we determine $n$ by sampling a binomial distribution $B(n,p)$, and then compare to $n^{\ast}_j$ to determine the payoff by \eq{constant_classical_oracle}. The sampling from a binomial distribution requires $O(\log n)$ expected samples and $O(\log n)$ expected time \cite{BKP15,FT15}. Thus, the cost of each iteration is $O(\log n)$. 

By Monte Carlo simulation, we repeat the above sampling procedure $O(\epsilon^{-2})$ times. Thus, we can estimate $\E[P]$ with the final complexity $O(\epsilon^{-2}\log n) = O(\epsilon^{-2}\log 1/\epsilon)$ in time.
\end{proof}

Similarly, we are given a zeroth-order quantum oracle by modifying \eq{oracle_BS} to be
\begin{equation}
U_{\psi}(|S_0\rangle|k\rangle|0\rangle) = |S_0\rangle|k\rangle|\psi(U^kD^{n-k}S_0)\rangle.
\label{eq:constant_quantum_oracle}
\end{equation}
where $\psi$ is the final payoff \eq{constant}. The same as the procedure of \eq{constant_classical_oracle}, the complexity of \eq{constant_quantum_oracle} is independent on $N$. 

We follow Theorem 22 of \cite{LMS20}, in which we replace the quantum walk by BOPM.

\begin{theorem}\label{thm:QBOPM}
We consider \prb{problem}  given by a Geometric Brownian Motion \eq{BSM} with an initial condition $S_0$. We are given the zeroth-order quantum oracle \eq{constant_quantum_oracle} for a final piecewise constant payoff as defined in \eq{constant}, which has a bounded variance independent of $h$. There exists a quantum algorithm that estimates $\E[P]$ up to additive error $\epsilon$ with probability at least 0.99 in cost 
\begin{equation}
\widetilde O(\epsilon^{-1}).
\end{equation}
\end{theorem}

\begin{proof}
Inspired by Theorem 22 of \cite{LMS20}, we apply QA-MC to the random seed used as input to a procedure for sampling from the binomial distributions. According to \prop{BOPM}, the cost of each iteration is $O(\log n)$ with $n = O(1/\epsilon^2)$, and amplitude estimation gives the final complexity $O(\epsilon^{-1}(\log 1/\epsilon)^{3/2} (\log\log 1/\epsilon)^2)$ in time.
\end{proof}

\section{Discussion}
\label{sec:discussion}

We have presented quantum-accelerated multilevel Monte Carlo methods for stochastic processes. We apply our algorithm to several applications arising in mathematical finance, in which we classify different financial models corresponding with different payoffs in detail. We have shown a quadratically improved dependence on precision can be achieved by our algorithm.

This work raises several natural open problems. First, from the PDE perspective, we only deal with parabolic PDEs as an application of simulating SDEs. However, Poisson’s equation and elliptic PDEs can be solved by classical multilevel Monte Carlo methods \cite{CGST11}. Can we apply our algorithm to Poisson's equation, elliptic PDEs, or more general PDEs?

Second, we consider several types of financial models that can be presented as a SDE model \prb{problem}. However, we only consider time-independent payoffs $\P(X_T)$ that only rely on $X_T$. For more general time-dependent payoffs $\P(X_t)$ where $X_t$ is the stochastic path in time $t$, such as the Lookback option \cite{Gil15}, can we achieve such a quadratic speed-up for this general model? Furthermore, there are some financial models that can be solved by variational inequalities, such as American option. Can we develop corresponding quantum algorithms for such a generalization? 

Finally, we aim to output a classical value for estimating the mean of a payoff (\prb{problem}). Can we provide other meaningful characteristics of stochastic processes, or some processes beyond the SDE modelling by quantum computer? And can we find more practical quantum input-output models for potential applications in finance or other fields?

\section*{Acknowledgments}

The authors thank Andrew M.\ Childs, Lin Lin, and Nick Whiteley for valuable discussions and comments. JPL did part of this work while visiting the Simons Institute for the
Theory of Computing in Berkeley and gratefully acknowledge its hospitality. 
We thank the National Energy Research Scientific Computing (NERSC)
center and the Berkeley Research Computing (BRC) program at the
University of California, Berkeley for making computational resources
available. 

The authors acknowledge support from National Science Foundation grant CCF-1813814, the U.S. Department of Energy, Office of Science, Office of Advanced Scientific Computing Research, Quantum Algorithms Teams and Accelerated Research in Quantum Computing programs. 
This work was also partially supported by the Department of Energy under Grant No. DE-SC0017867 and the National Science Foundation under the Quantum Leap
Challenge Institutes (D.A.,J.W.).  
We acknowledge support from the QuantERA ERA-NET Cofund in Quantum Technologies implemented within the European Union’s Horizon 2020 Programme (QuantAlgo project) and EPSRC grants EP/R043957/1 and EP/T001062/1. This project has received funding from the European Research Council (ERC) under the European Union’s Horizon 2020 research and innovation programme (grant agreement No.\ 817581).

\bibliographystyle{abbrvurl}
\bibliography{qparabolic}

\newpage
\appendix

\section{Different types of errors}
\label{app:errors}

Let a random variable $Y$ be an estimator of some unknown quantity $a$. 
In this paper we have considered two different types of errors, 
mean-squared error $\E[(Y-a)^2]$ and additive error $|Y-a|$. 
Here we will show that these two types of errors are indeed almost  equivalent. 
More precisely, we will show that, up to some absolute pre-constants in the errors and logarithmic factors in the cost, the situation that mean-squared error is on the level of $\epsilon^2$ indicates that the additive error is on the level of $\epsilon$ with probability at least 0.99, and \emph{vice versa}. 

\begin{proposition}
    Let $\mathcal{A}$ be a (classical or quantum) algorithm that generates a random variable $Y$ to estimate some unknown quantity $a$. Assuming that $\E[|Y|^{2+\delta}] < \infty$ for some $\delta > 0$, we have:  
    \begin{enumerate}
        \item If $\E[(Y-a)^2] \leq \epsilon^2$, then there exists an algorithm which repeats $\mathcal{A}$ a constant number of times and outputs $\widehat{Y}$ such that $|\widehat{Y}-a| \leq 3\epsilon$ with probability at least 0.99.
        \item If $|Y-a| \leq \epsilon$ holds with probability at least 0.99, then there exists an algorithm which repeats $\mathcal{A}$ $O(\log(1/\epsilon))$ times and outputs $\widehat{Y}$ such that $\E[(\widehat{Y}-a)^2] \leq 2\epsilon^2$.
    \end{enumerate}
\end{proposition}
\begin{proof}
    1. By $\E[(Y-a)^2] \leq \epsilon^2$ and the bias-variance decomposition
    \begin{equation}
        \E[(Y-a)^2] = \E [(Y-\E [Y])^2] + (\E [Y] - a)^2,
    \end{equation}
    we have $\E [(Y-\E [Y])^2] \leq \epsilon^2$ and $|\E [Y] - a| \leq \epsilon$. 
    Since 
    \begin{equation}
        |Y-\E [Y]| \geq  |Y- a| -  |a-\E [Y]| \geq |Y- a| - \epsilon, 
    \end{equation}
    the event $\{|Y-a| \geq 3\epsilon\}$ is a subset of the event $\{|Y-\E [Y]| \geq 2\epsilon\}$. 
    Together with Chebyshev's inequality, we have 
    \begin{equation}
    \begin{split}
        \mathbb{P}(|Y-a| \geq 3\epsilon) &\leq \mathbb{P}(|Y-\E [Y]| \geq 2\epsilon) \\
        &\leq \frac{\E [(Y-\E [Y])^2] }{(2\epsilon)^2} \\
        & \leq \frac{1}{4}.
    \end{split}
    \end{equation}
    This indicates that a single sample of $Y$ can estimate $a$ up to additive error $3\epsilon$ with probability at least $3/4$. 
    Then by \lem{powering lemma}, it suffices to repeat $\mathcal{A}$ a constant number of times to boost the success probability to $0.99$. 
    This completes the proof of the first part. 
    
    2. First, according to \lem{powering lemma}, we can construct an estimator $\widehat{Y}$ such that $|\widehat{Y} - a| \leq \epsilon$ with probability at least
    $$1-2^{-2-2/\delta}\left(\E |Y|^{2+\delta}+|a|^{2+\delta}\right)^{-2/\delta}\epsilon^{2+4/\delta} $$ 
    by repeating the algorithm $\mathcal{A}$ $O(\log(1/\epsilon))$ times. 
    Let $\psi_A$ be the indicator function of a set $A$, \emph{i.e. }$\psi_A(x) = 1$ if $x \in A$ and $\psi_A(x) = 0$ if $x \notin A$. 
    We split the mean-squared error into two parts, 
    \begin{equation}
        \begin{split}
            \E[(\widehat{Y}-a)^2] &= \int (y-a)^2 \d \mathbb{P}_{\widehat{Y}}(y) \\
            &= \int (y-a)^2 \psi_{\{|y-a| < \epsilon\}}(y) \d \mathbb{P}_{\widehat{Y}}(y) + \int (y-a)^2 \psi_{\{|y-a| \geq \epsilon\}}(y) \d \mathbb{P}_{\widehat{Y}}(y). 
        \end{split}
    \end{equation}
    The first integral can be bounded as 
    \begin{equation}
        \int (y-a)^2 \psi_{\{|y-a| < \epsilon\}}(y) \d \mathbb{P}_{\widehat{Y}}(y) \leq \int \epsilon^2 \psi_{\{|y-a| < \epsilon\}}(y) \d \mathbb{P}_{\widehat{Y}}(y) 
        \leq \int \epsilon^2 \d \mathbb{P}_{\widehat{Y}}(y) =\epsilon^2.
    \end{equation}
    By H\"older's inequality, the second integral can be bounded as 
    \begin{equation}
    \begin{split}
         & \quad \int (y-a)^2 \psi_{\{|y-a| \geq \epsilon\}}(y) \d \mathbb{P}_{\widehat{Y}}(y) \\
        &\leq \left[\int |y-a|^{2+\delta}\d \mathbb{P}_{\widehat{Y}}(y) \right]^{\frac{2}{2+\delta}}\left[\int\psi_{\{|y-a| \geq \epsilon\}}(y)  \d \mathbb{P}_{\widehat{Y}}(y) \right]^{\frac{\delta}{2+\delta}} \\
        & = \left[\int |y-a|^{2+\delta}\d \mathbb{P}_{\widehat{Y}}(y) \right]^{\frac{2}{2+\delta}}\left[\mathbb{P}(|\widehat{Y}-a|\geq \epsilon) \right]^{\frac{\delta}{2+\delta}} \\
        & \leq 2^{-\frac{2+2\delta}{2+\delta}}\left(\E |Y|^{2+\delta}+|a|^{2+\delta}\right)^{-\frac{2}{2+\delta}}\epsilon^2\left[\int |y-a|^{2+\delta}\d \mathbb{P}_{\widehat{Y}}(y) \right]^{\frac{2}{2+\delta}}  \\
        & \leq 2^{-\frac{2+2\delta}{2+\delta}}\left(\E |Y|^{2+\delta}+|a|^{2+\delta}\right)^{-\frac{2}{2+\delta}}\epsilon^2 \left[\int (|y|+|a|)^{2+\delta}\d \mathbb{P}_{\widehat{Y}}(y) \right]^{\frac{2}{2+\delta}} \\
        & \leq \left(\E |Y|^{2+\delta}+|a|^{2+\delta}\right)^{-\frac{2}{2+\delta}}\epsilon^2 \left[\int |y|^{2+\delta}\d \mathbb{P}_{\widehat{Y}}(y) + \int |a|^{2+\delta}\d \mathbb{P}_{\widehat{Y}}(y) \right]^{\frac{2}{2+\delta}} \\
        & = \left(\E |Y|^{2+\delta}+|a|^{2+\delta}\right)^{-\frac{2}{2+\delta}}\epsilon^2  \left[\E |Y|^{2+\delta} + |a|^{2+\delta} \right]^{\frac{2}{2+\delta}}\\
        &= \epsilon^2 .
    \end{split}
    \end{equation}
    Therefore the mean-squared error $\E[(\widehat{Y}-a)^2]$ is bounded by $2\epsilon^2$. 
    This completes the proof of the second part. 
\end{proof}

\section{Numerical Results}
\label{app:numerical results}
In this part, we test several numerical schemes of SDE on Black-Scholes option pricing model to obtain a classical estimate of the parameters $\alpha$ and $\beta$ in  \prop{alpha_beta_gamma_general}. We consider five numerical schemes, including Euler-Maruyama scheme \eq{Euler}, Milstein scheme \eq{Milstein}, strong order 1.5 scheme\eq{strong1.5}, strong order 2 scheme\eq{strong2} \eq{strong2_2} and strong order 3 scheme\eq{strong3}. The strong order 1.5 scheme is based on 
Taylor-It\^{o} expansion
\begin{equation}\label{eq:strong1.5}
\begin{split}
    \widehat{X}_{k+1} &= \widehat{X}_k + \mu_k h + \sigma_k \Delta W_k +\sigma_k \sigma_k' I_{(1,1)}+\sigma_k\mu_k'I_{(1,0)}+(\mu_k\mu_k' +\frac{1}{2}\sigma_k^2 \mu_k) \frac{h^2}{2}\\
    &\quad +(\mu_k\sigma_k'+\frac{1}{2}\sigma_k^2\sigma_k'')I_{(0,1)}+\sigma_k(\sigma_k\sigma_k''+(\sigma_k')^2)I_{(1,1,1)}.
\end{split}
\end{equation}
Here, \begin{equation}
    \begin{split}
        &\mu_k:=\mu(\widehat{X}_k,t),\qquad\mu_k':=\partial_X \mu(\widehat{X}_k,t),\\
        &\sigma_k:=\sigma(\widehat{X}_k,t),\qquad \sigma_k':=\partial_X \sigma(\widehat{X}_k,t),\qquad\sigma_k'':=\partial_X^2 \sigma(\widehat{X}_k,t).
    \end{split}
\end{equation}
Without other notice, for any function $f:=f(x,t)$, $f'$ denote the partial derivative of function $f$ with respect to $x$. $I_{(1,1)}, I_{(1,0)}, I_{(0,1)}$ and $I_{(1,1,1)}$ are multiple It\^{o} integrals:

\begin{equation}
    \begin{split}
        I_{(1)} &=\int_{t_n}^{t_{n+1}} \d W_{s_1} =\Delta W_k,\\
        I_{(1,1)}&=\int_{t_n}^{t_{n+1}}  \int_{t_n}^{s_2} \d W_{s_1} \d W_{s_2}=\frac{1}{2}(I_{(1)}^2-h),\\
        I_{(1,0)}&=\int_{t_n}^{t_{n+1}}  \int_{t_n}^{s_2} \d W_{s_1} ds_2,\\
        I_{(0,1)}&=\int_{t_n}^{t_{n+1}}  \int_{t_n}^{s_2} ds_1 \d W_{s_2}=h I_{(1)}-I_{(1,0)},\\
        I_{(1,1,1)}&=\int_{t_n}^{t_{n+1}}  \int_{t_n}^{s_3} \int_{t_n}^{s_2} \d W_{s_1} \d W_{s_2} \d W_{s_3}=\frac{1}{2}(\frac{1}{3}I_{(1)}^2-h)I_{(1)}.\\
    \end{split}
\end{equation}
In practice, we use two independent $\mathcal{N}(0,1)$ distributed random variables $U_1$ and $U_2$ to approximate $I_{(1)}$ and $I_{(1,0)}$ with
\begin{equation}
    \begin{split}
         I_{(1)} =\sqrt{h}U_1,\qquad I_{(1,0)}=\frac{1}{2}h^{\frac{3}{2}}(U_1+\frac{1}{\sqrt{3}}U_2).
    \end{split}
\end{equation}

The integer strong order Taylor schemes can be conveniently derived form a Taylor-Stratonovich expansion, and thus we use the following strong order 2 scheme \cite{Kuz18}: 
\begin{equation}\label{eq:strong2}
\begin{split}
    \widehat{X}_{k+1} &= \widehat{X}_k + \bar{\mu}_k h + \sigma_k \Delta W_k +\sigma_k \sigma_k' J_{(1,1)}+\sigma_k \bar{\mu}_k' J_{(1,0)}+ \bar{\mu}_k\sigma_k' J_{(0,1)}+ \bar{\mu}_k \bar{\mu}_k' \frac{h^2}{2}\\
    &\quad +\sigma_k(\sigma_k\sigma_k')'J_{(1,1,1)}+ \bar{\mu}_k (\sigma_k\sigma_k')' J_{(0,1,1)}+\sigma_k( \bar{\mu}_k \sigma_k')' J_{(1,0,1)}+\sigma_k(\sigma_k \bar{\mu}_k')'J_{(1,1,0)}\\
    &\quad +\sigma_k(\sigma_k (\sigma_k \sigma_k')')' J_{(1,1,1,1)},
\end{split}
\end{equation}
where 
\begin{equation}
     \bar{\mu}=\mu-\frac{1}{2}\sigma \sigma',
\end{equation}
and $J_{(1,1)}, J_{(1,0)}, J_{(0,1)}, J_{(1,1,1)}, J_{(1,1,0)}, J_{(1,0,1)}, J_{(0,1,1)}$ and $J_{(1,1,1,1)}$ are multiple Stratonovich integrals:
\begin{equation}
    \begin{split}
        J_{(1)} &=\int_{t_n}^{t_{n+1}} \circ \d W_{s_1} = I_{(1)}=\Delta W_k,\\
        J_{(1,1)}&=\int_{t_n}^{t_{n+1}}  \int_{t_n}^{s_2}\circ \d W_{s_1} \circ \d W_{s_2}=\frac{1}{2!} (J_{(1)})^2,\\
        J_{(1,0)}&=\int_{t_n}^{t_{n+1}}  \int_{t_n}^{s_2} \circ \d W_{s_1}  ds_2,\\
        J_{(0,1)}&=\int_{t_n}^{t_{n+1}}  \int_{t_n}^{s_2}  ds_1 \circ \d W_{s_2},\\
        J_{(1,1,1)}&=\int_{t_n}^{t_{n+1}}  \int_{t_n}^{s_3} \int_{t_n}^{s_2}\circ \d W_{s_1}\circ \d W_{s_2} \circ \d W_{s_3}=\frac{1}{3!} (J_{(1)})^3,\\
        J_{(1,1,0)}&=\int_{t_n}^{t_{n+1}}  \int_{t_n}^{s_3} \int_{t_n}^{s_2} ds_1\circ \d W_{s_2}\circ \d W_{s_3},\\
        J_{(1,0,1)}&=\int_{t_n}^{t_{n+1}}  \int_{t_n}^{s_3} \int_{t_n}^{s_2}\circ \d W_{s_1} ds_2\circ \d W_{s_3},\\
    \end{split}
\end{equation}

\begin{equation*}
    \begin{split}
        J_{(0,1,1)}&=\int_{t_n}^{t_{n+1}}  \int_{t_n}^{s_3} \int_{t_n}^{s_2}\circ \d W_{s_1}\circ \d W_{s_2}  ds_3,\\
        J_{(1,1,1,1)}&=\int_{t_n}^{t_{n+1}} \int_{t_n}^{s_3} \int_{t_n}^{s_3} \int_{t_n}^{s_2}\circ \d W_{s_1}\circ \d W_{s_2}\circ \d W_{s_3}\circ \d W_{s_4}=\frac{1}{4!} (J_{(1)})^4.\\
    \end{split}
\end{equation*}
We remark that there also exists the strong order 2 scheme derived from Taylor-It\^{o} expansion~\cite{KP13}. Typically, we approximate multiple Stratonovich integrals and multiple It\^{o} integrals by introducing additional random variables. However, the approximation of multiple Stratonovich integrals is simpler and requires minimal set of random variables. That is also why we prefer high order numerical schemes derived from Taylor-Stratonovich expansion, such as the strong order 2 scheme and strong order 3 scheme in the following. 

Since we consider Black-Scholes model,
\begin{equation}\label{eq:BSMgeneral}
\d{S_t} = \mu S_t\d t + \sigma S_t\d W_t,
\end{equation} 
the strong order 2 scheme can be simplified to 
\begin{equation}\label{eq:strong2_2}
    \begin{split}
         \widehat{X}_{k+1} &= \widehat{X}_k + \bar{\mu}\widehat{X}_k h + \sigma \widehat{X}_k \Delta W_k +\sigma^2 \widehat{X}_k J_{(1,1)}+\bar{\mu}\sigma \widehat{X}_k h\Delta W_k+ \bar{\mu}^2 \widehat{X}_k \frac{h^2}{2}\\
    &\quad +\sigma^3 \widehat{X}_k J_{(1,1,1)}+ \bar{\mu}\sigma^2 \widehat{X}_k h J_{(1,1)}+\sigma^4 \widehat{X}_k J_{(1,1,1,1)},
    \end{split}
\end{equation}
where $\bar{\mu}=\mu-\frac{1}{2}\sigma^2$ and we use a nice property of multiple Stratonovich integral that
\begin{equation}
    J_{(1,0)}+J_{(0,1)}=h J_{(1)},\quad \quad J_{(1,1,0)}+J_{(1,0,1)}+ J_{(0,1,1)}= h J_{(1,1)}.
\end{equation}

The strong order 3 schemes for general SDEs are thoroughly discussed in \cite{Kuz18}. We choose the strong order 3.0 scheme based on Taylor-Stratonovich expansion and the corresponding discretization scheme for \eq{BSMgeneral} is as follow.
\begin{equation}\label{eq:strong3}
\begin{split}
    \widehat{X}_{k+1}&=\widehat{X}_k+\bar{\mu} \widehat{X}_k h+\sigma \widehat{X}_k \Delta W_k + \sigma^2 \widehat{X}_k J_{(1,1)}+\bar{\mu} \sigma \widehat{X}_k h \Delta W_k +\bar{\mu}^2 \widehat{X}_k\frac{h^2}{2}+\sigma^3 \widehat{X}_k J_{(1,1,1)} \\
    &\quad +\bar{\mu} \sigma^2 \widehat{X}_k h J_{(1,1)} + \sigma^4 \widehat{X}_k J_{(1,1,1,1)}+ \bar{\mu}^2\sigma\widehat{X}_k  \frac{h^2}{2}\Delta W_k+\bar{\mu} \sigma^3 \widehat{X}_k h J_{(1,1,1)}+ \sigma^5 \widehat{X}_k J_{(1,1,1,1,1)}\\
    & \quad + \bar{\mu}^3\widehat{X}_n\frac{h^3}{6} +\sigma^6 \widehat{X}_k J_{(1,1,1,1,1,1)}+\bar{\mu}^2 \sigma^2\widehat{X}_k\frac{h^2}{2} J_{(1,1)}+\bar{\mu}\sigma^4 \widehat{X}_k h J_{(1,1,1,1)}
\end{split}
\end{equation}
where $J_{(1,1,1,1,1)}$ and $J_{(1,1,1,1,1,1)}$ are multiple Stratonovich integrals with
\begin{equation}
    J_{(1,1,1,1,1)}=\frac{J_{(1)}}{5!}, \qquad J_{(1,1,1,1,1,1)}=\frac{J_{(1)}}{6!}. 
\end{equation}

As for options, we consider European option and Digital option pricing. For European option, the payoff function is 
 \begin{equation}
     \psi (S_T) = \exp (-\mu T) \max\{S_T-K,0\}
 \end{equation}
 For Digital option, the payoff function is 
 \begin{equation}
     \psi (S_T)= 5\exp(-\mu T)(1+ \mathcal{H}(S_T-K))
 \end{equation}
 where $\mathcal{H}$ is Heaviside step function.
 
 \tab{tab2}, \ref{tab:tab3} and \fig{fig1}, \ref{fig:fig2}, \ref{fig:fig3}, \ref{fig:fig4}, \ref{fig:fig5} show our numerical results on estimating $\alpha$ and $\beta$ as well as the convergence performance of the errors at each layer of MLMC. 
 The five schemes mentioned above are tested by running respectivly $10^6$, $10^6$, $10^7$, $10^8$, $10^9$ independent simulations with the following choice of parameters: 
 \begin{equation}
     \mu=0.05,\qquad \sigma=0.2,\qquad T=1,\qquad S_0=100,\qquad K=100.
 \end{equation}
 The numerical scalings of $\alpha$ and $\beta$ are estimated using linear regression based on the last four points, in which the time step size $h$ is small enough such that the corresponding numerical scheme has already well converged (as shown in the figures) and the numerical errors are dominated by the leading order term. 
 
 \begin{table}[ht]
    \centering
\begin{tabular}{|c|c|c|c|c|c|}
\hline  
Option &Euler Maruyama & Milstein & strong order 1.5  & strong order 2 & strong order 3\\
\hline 
European & 0.976999  & 1.962848 & 2.970166 & 3.964626 & 5.958417\\
\hline 
Digital & 0.473426  & 0.869393  & 1.452448  & 1.775679  & NAN\\
\hline
\end{tabular}
    \caption{numerical estimates of $\beta$ based on linear regression for five schemes}
    \label{tab:tab2}
\end{table}

\begin{table}[ht]
    \centering
\begin{tabular}{|c|c|c|c|c|c|}
\hline  
Option &Euler Maruyama & Milstein & strong order 1.5  & strong order 2 & strong order 3 \\
\hline 
European & 1.136214 & 0.979572  & 1.747239  & 1.970829 & 2.961041 \\
\hline 
Digital & 1.023176 & 0.791818  & 1.853158 & 1.827618 & NAN\\
\hline
\end{tabular}
    \caption{numerical estimates of $\alpha$ based on linear regression for five schemes}
    \label{tab:tab3}
\end{table}
 
 \tab{tab2} and \tab{tab3} shows the estimate of $\beta$ and $\alpha$ of these schemes respectively. 
 In the case of European option, the estimates of $\beta$ for five schemes all approximately equal twice the corresponding strong order, which implies that our theoretical estimate for $\beta$ in \prop{alpha_beta_gamma_general} is sharp for Lipschitz continuous functions. 
 The estimates of $\alpha$ with higher order schemes are approximately equal to the corresponding strong order, which also agrees well with our theoretical estimate for $\alpha$. 
 In the case of Digital option, other than the strong order 3 scheme (which will be discussed later), the estimates of $\beta$ are roughly equal to the corresponding strong order of the schemes, which again verifies our theoretical results for non-Lipschitz payoff functions. 
 We remark that, similarly to the European option case, some of the estimates for $\alpha$ are larger than the strong order. 
 This is because our theoretical estimate for $\alpha$ is quite conservative that we only employ properties of strong order, while the convergence of payoff function is more related to the weak order, which describes the convergence order of the moments of the stochastic process. 
 This also implies that our theoretical estimate for $\alpha$ might be improvable for particular schemes and payoff functions. 

 We notice that for Digital option, strong order 3 scheme cannot output the estimate of $\alpha$ and $\beta$. That is because strong order 3 scheme is of pretty high accuracy. When we compute the option price, since the expectation of $S_T$ under our parameter choices is larger than $K$, the Heaviside step function $\mathcal{H}$ outputs 1 with extremely high probability and the option price is nearly  deterministic. The outcome of numerical experiment has such a small variance that the influence of the increase in level $l$ on the mean and variance is hard to distinguish, which explains the results shown in \fig{fig5}. 

 So we make some changes to the choice of parameters. We increase $\sigma$ to 1.5, which increase the randomness, and choose $K$ to be  exactly the theoretical expectation $100e^{-0.05}$. The numerical results are shown in \fig{fig6}. The estimate of $\beta$ is 2.957982, which agrees well with theoretical result, and the estimate of $\alpha$ is 2.438328, which might be estimated more accurately by using more samplings as well as smaller time step sizes. 

\begin{figure}[htbp]
        \centering
        \includegraphics[scale=0.65]{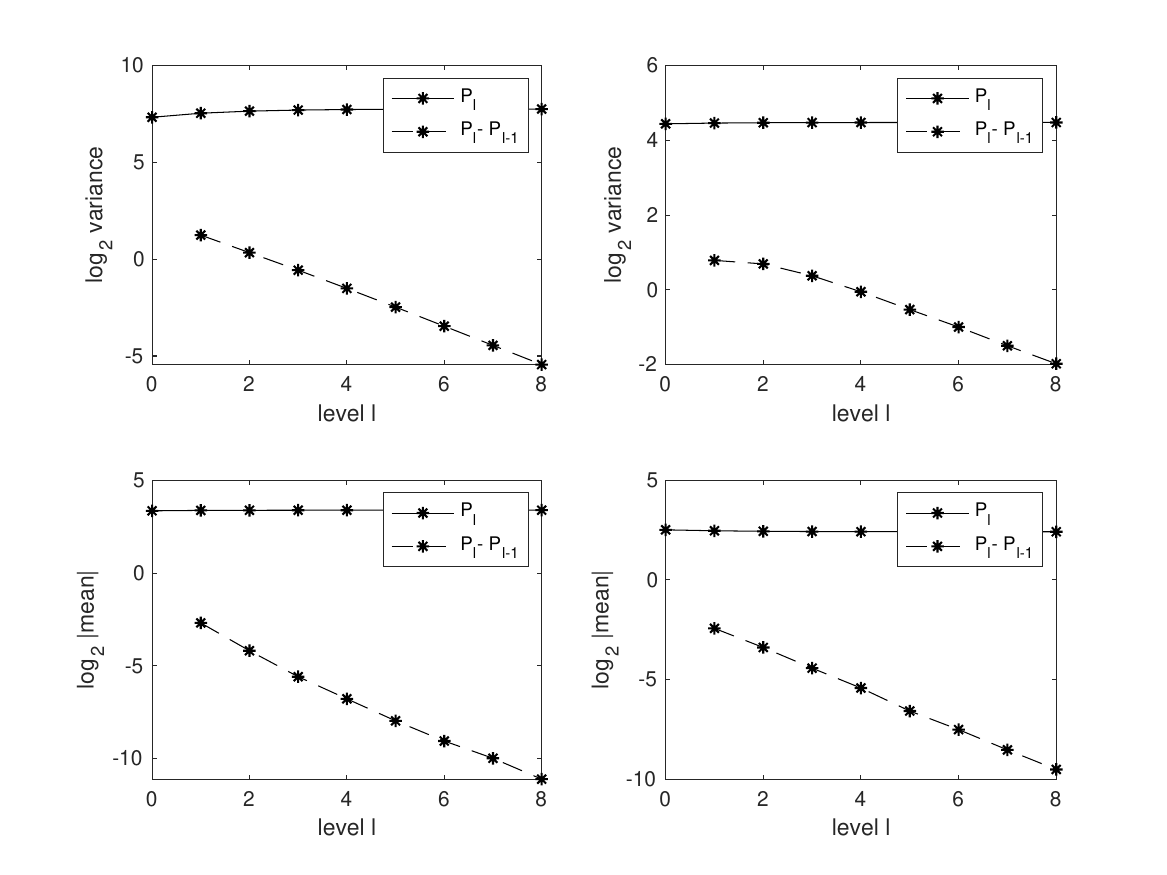}
        \caption{estimates of $\beta$ and $\alpha$ when using Euler Maruyama Scheme for European option (left) and Digital option (right). The top two plots show the variance versus level $l$, and the slope gives an estimate of $\beta$. The bottom two plots show the mean versus level $l$ and the slope gives an estimate of $\alpha$. }
        \label{fig:fig1}
\end{figure}
\begin{figure}[htbp]
       \centering
        \includegraphics[scale=0.65]{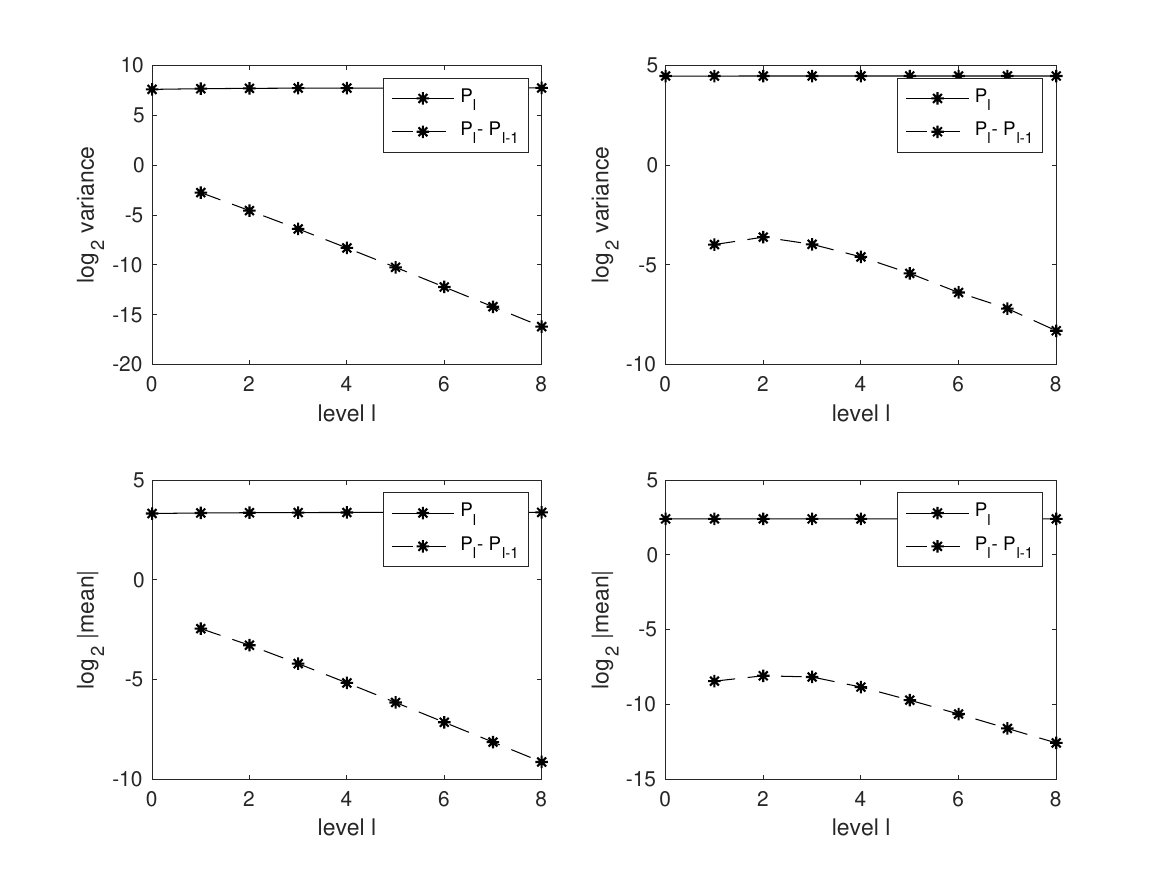}
        \caption{estimates of $\beta$ and $\alpha$ when using Milstein Scheme for European option (left) and Digital option (right). The top two plots show the variance versus level $l$, and the slope gives an estimate of $\beta$. The bottom two plots show the mean versus level $l$ and the slope gives an estimate of $\alpha$.}
        \label{fig:fig2}
\end{figure}
\begin{figure}[htbp]
       \centering
        \includegraphics[scale=0.65]{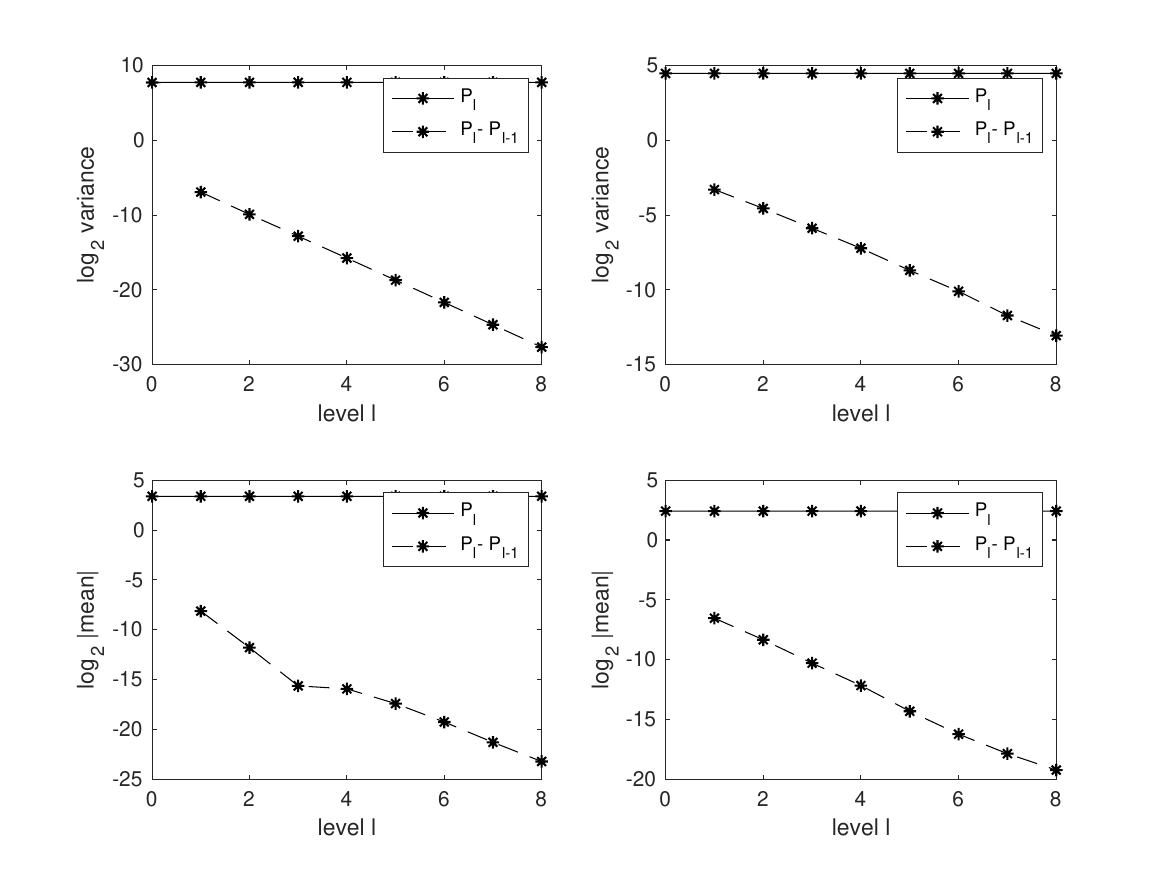}
        \caption{estimates of $\beta$ and $\alpha$ when using Strong order 1.5 Scheme for European option (left) and Digital option (right). The top two plots show the variance versus level $l$, and the slope gives an estimate of $\beta$. The bottom two plots show the mean versus level $l$ and the slope gives an estimate of $\alpha$.}
        \label{fig:fig3}
\end{figure}
\begin{figure}[htbp]
        \centering
        \includegraphics[scale=0.65]{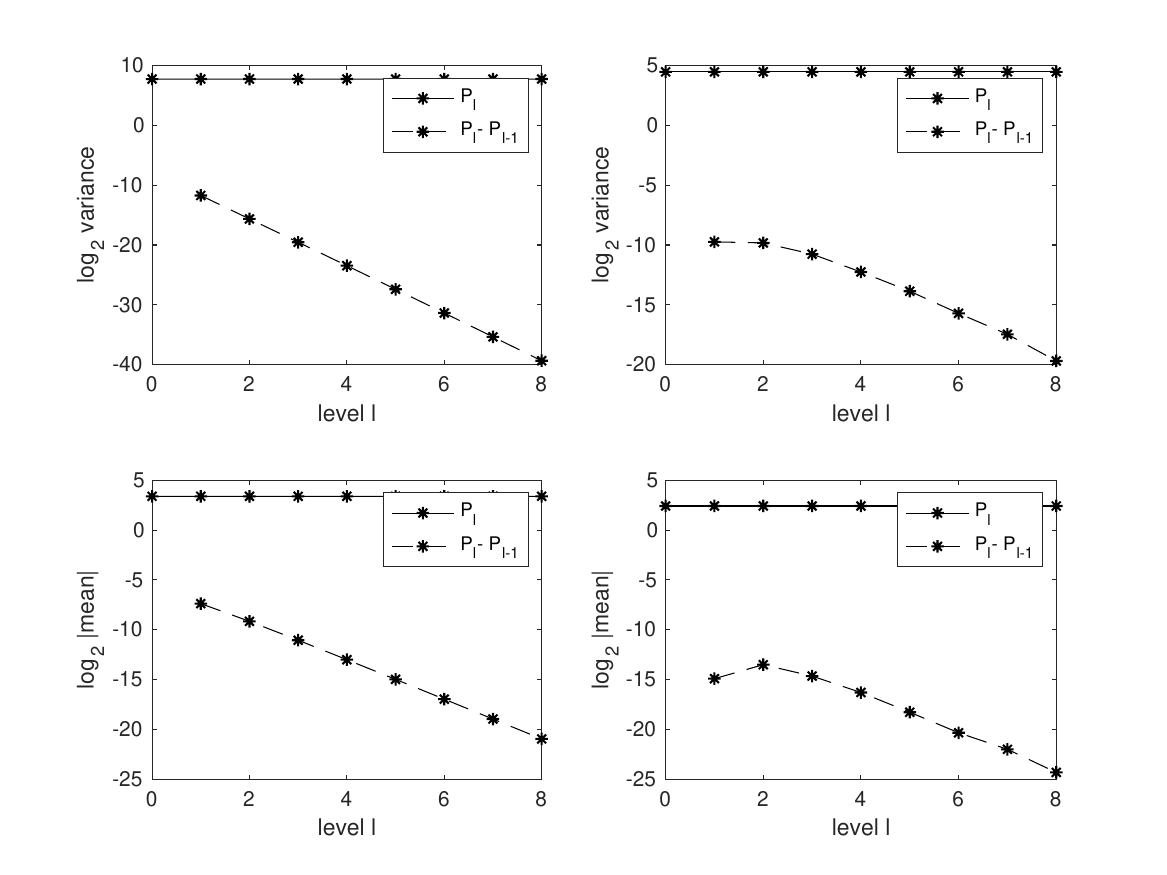}
        \caption{estimates of $\beta$ and $\alpha$ when using Strong order 2 Scheme for European option (left) and Digital option (right). The top two plots show the variance versus level $l$, and the slope gives an estimate of $\beta$. The bottom two plots show the mean versus level $l$ and the slope gives an estimate of $\alpha$.}
        \label{fig:fig4}
\end{figure}
\begin{figure}[htbp]
        \centering
        \includegraphics[scale=0.65]{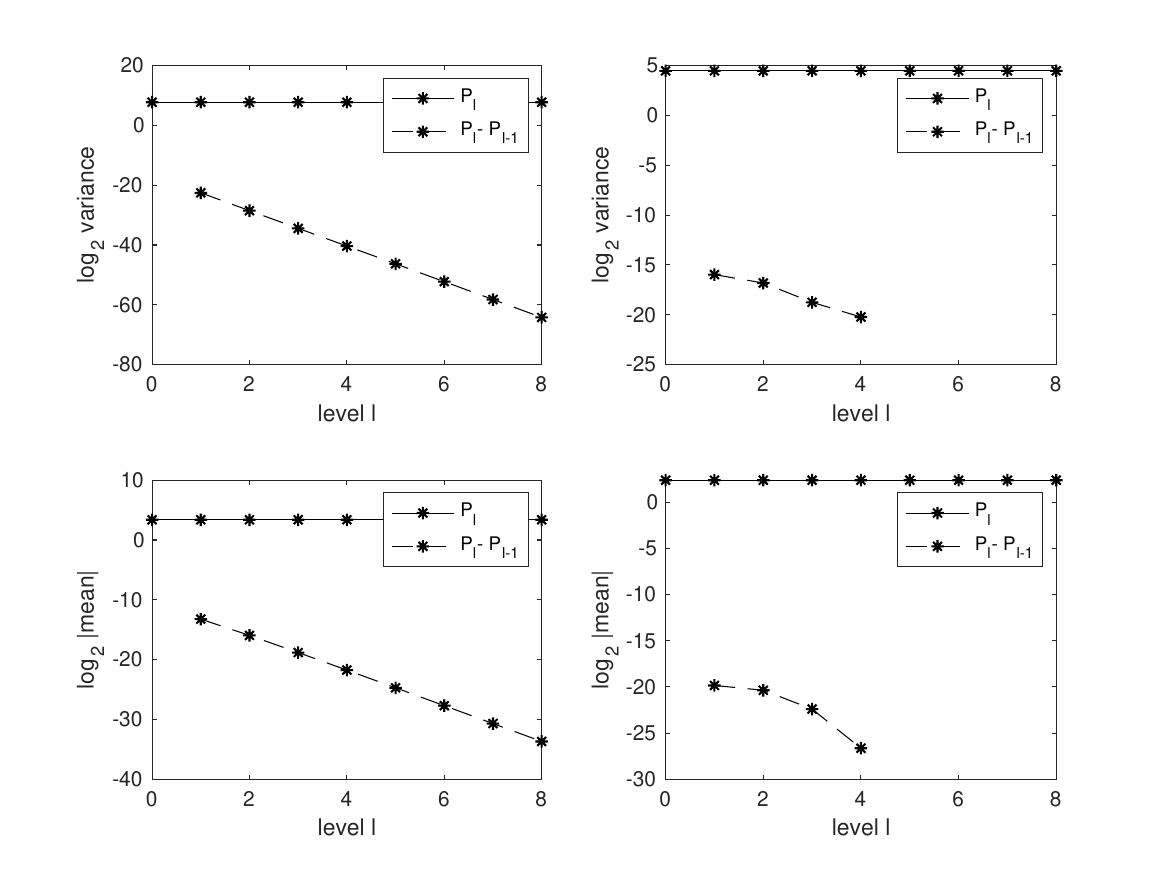}
        \caption{estimates of $\beta$ and $\alpha$ when using Strong order 3 Scheme for European option (left) and Digital option (right). The top two plots show the variance versus level $l$, and the slope gives an estimate of $\beta$. The bottom two plots show the mean versus level $l$ and the slope gives an estimate of $\alpha$.}
        \label{fig:fig5}
\end{figure}

\begin{figure}[htbp]
        \centering
        \includegraphics[scale=0.54]{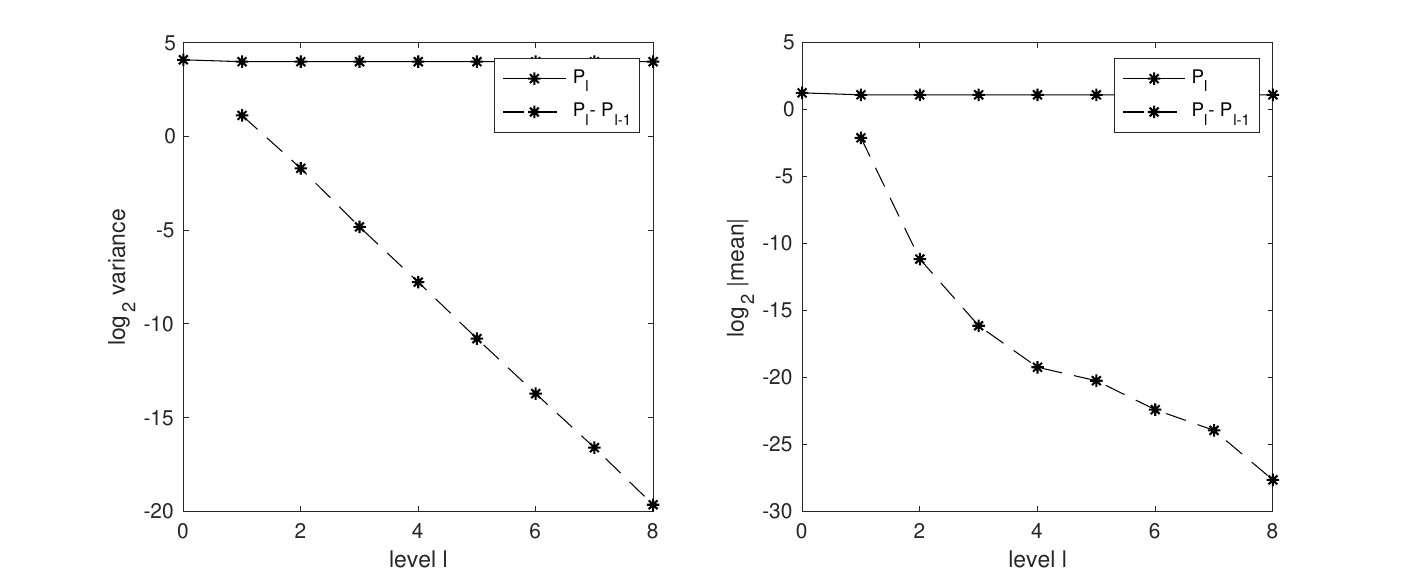}
        \caption{estimates of $\beta$ and alpha when using Strong order 3 Scheme for Digital option with new parameters $\sigma=1.5$ and $K=100e^{-0.05}$. The left plot shows the variance versus level $l$, and the slope gives an estimate of $\beta$. The right plot shows the mean versus level $l$ and the slope gives an estimate of $\alpha$.}
        \label{fig:fig6}
\end{figure}

\end{document}